\documentclass{article}

\usepackage[english]{babel}
\usepackage[utf8]{inputenc}
\usepackage[T1]{fontenc}
\usepackage{calligra}
\usepackage{amsfonts,amsmath,amssymb}
\usepackage{amsthm}
\usepackage{mathabx}
\usepackage[usenames]{color} 
\usepackage{xspace}
\usepackage[svgnames]{xcolor}
\usepackage{subfig,caption}
\usepackage{todonotes}
\usepackage{paralist}
\usepackage{microtype}
\usepackage{tikz}
\usetikzlibrary{shapes}
\usetikzlibrary{calc,automata}
\usepackage{subfig}
\usepackage{bookman}
\usepackage{authblk}
\usepackage[margin = 2.5cm]{geometry}

\newtheorem{theorem}{Theorem}
\newtheorem{example}{Example}
\newtheorem{proposition}{Proposition}
\newtheorem{lemma}{Lemma}
\newtheorem{corollary}{Corollary}

\newtheorem{remark}{Remark}

\makeatletter
\def\timenow{\@tempcnta\time
  \@tempcntb\@tempcnta
  \divide\@tempcntb60
  \ifnum10>\@tempcntb0\fi\number\@tempcntb
  \multiply\@tempcntb60
  \advance\@tempcnta-\@tempcntb
  :\ifnum10>\@tempcnta0\fi\number\@tempcnta}
\makeatother

\newcommand{\eat}[1]{}

\renewcommand{\epsilon}{\varepsilon}
\renewcommand{\phi}{\varphi}

\newcommand{\defin}[1]{\emph{\textbf{#1}}}

\newcommand{\distribution}[1]{\mathcal{D}(#1)}

\newcommand{\arena}{\mathcal{A}}

\newcommand{\game}{\mathbb{G}}

\newcommand{\Eve}{Eve\xspace}
\newcommand{\Adam}{Adam\xspace}

\newcommand{\Evei}{E}
\newcommand{\Adami}{A}
\newcommand{\Ei}{E}
\newcommand{\Ai}{A}
\newcommand{\states}{S}
\newcommand{\actions}{\Sigma}
\newcommand{\action}{\sigma}
\newcommand{\ftrans}{\delta}

\newcommand{\requiv}{\sim}
\newcommand{\rdoubleequiv}{\approx}

\newcommand{\any}{*}
\newcommand{\outcomes}[3]{\mathrm{Outcomes}(#1,#2,#3)}

\newcommand{\cone}[1]{\mathrm{cone}(#1)}

\newcommand{\proba}[3]{\mathrm{Pr}_{#1}^{#2,#3}}

\newcommand{\objective}{\mathcal{O}}

\newcommand{\fstates}{F}
\newcommand{\updateB}[1]{\mathrm{UpBelief}_{\!#1}}
\newcommand{\belief}[2]{\mathrm{Belief}_{#2}^{#1}}

\newcommand{\wbelief}[1]{\mathcal{B}^{{AS}}_{#1}}

\newcommand{\play}{\lambda}

\newcommand{\exptime}{\textsc{ExpTime}\xspace}

\newcommand{\BeliefOperator}{\Xi}

\newcommand{\win}{q_f}
\newcommand{\wait}{q_w}
\newcommand{\head}{0}
\newcommand{\tail}{1}
\newcommand{\barre}{\!\mid\!}

\newcommand{\resp}{{resp.}\xspace}
\newcommand{\Reach}{\mathrm{Reach}}
\newcommand{\ReachN}{\Reach^{\leq n}}
\newcommand{\ReachNP}{\Reach^{\leq n-1}}

\newcommand{\ie}{\emph{i.e.} }
\newcommand{\deltaT}{\delta_{\mathcal{T}}}
\newcommand{\deltaTp}{\delta_{\mathcal{T}'}}
\newcommand{\act}{Act}
\newcommand{\actT}{{\act_{\mathcal{T}}}}
\newcommand{\strat}{\phi}

\newcommand{\trans}{\mathcal{T}}

\newcommand\os[1]{}
\newcommand\ac[1]{}
\newcommand\cl[1]{}
\newcommand\review[1]{}
\newcommand\oschanged[1]{#1}

\newcommand\vlong[1]{}
\newcommand\reviewP[2]{}
\newcommand\reviewW[2]{}
\newcommand\reviewF[2]{}
\newcommand\answer[1]{}
\newcommand{\remove}[1]{}

\begin{document}



\title{Pure Strategies in Imperfect Information Stochastic Games}

\author[1]{Arnaud Carayol\thanks{{Arnaud.Carayol@univ-mlv.fr}}}
\author[2]{Christof L\"oding\thanks{{loeding{@}informatik.rwth-aachen.de}}}
\author[3]{Olivier Serre\thanks{{Olivier.Serre@cnrs.fr}}}
\affil[1]{LIGM (CNRS \& Université Paris Est)}
\affil[2]{Informatik 7,  RWTH Aachen, Germany}
\affil[3]{IRIF (CNRS \& Université Paris Diderot -- Paris 7)}
\date{}

\maketitle


\begin{abstract}
We consider imperfect information stochastic games where we require the players to use pure (\emph{i.e.} non randomised) strategies. We consider reachability, safety, B\"uchi and co-B\"uchi objectives, and investigate the existence of almost-sure/positively winning strategies for the first player when the second player is perfectly informed or more informed than the first player. We obtain decidability results for positive reachability and almost-sure B\"uchi with optimal algorithms to decide existence of a pure winning strategy and to compute one if it exists. We complete the picture by showing that positive safety is undecidable when restricting to pure strategies even if the second player is perfectly informed.
\end{abstract}



\section{Introduction}

The study of two-player games has received a lot of attention in the last decade, mainly motivated by applications to the verification of reactive open systems. Those systems are composed of a program (represented by the first player, \Eve) and some (possibly hostile) environment (represented by the second player, \Adam). The verification problem consists in deciding whether the program can be restricted so that the system meets some given specification \emph{whatever} the environment does. Here, restricting the program means synthesizing a controller~\cite{RW87}, which, in terms of games, is equivalent to designing a strategy for \Eve that is winning against \emph{any} strategy of \Adam.

Of course, the class of games to consider depends on the class of systems that one intends to model. This may lead to consider various features such as concurrency (the players \emph{independently} and \emph{simultaneously} choose their action, whose \emph{parallel} execution determines the next state), stochastic transitions (the next state is chosen according to a probability distribution depending on the current state and on the actions chosen by the players) or imperfect information (the players do not observe the exact state). Note that imperfect information is necessary if one wants for instance to model a system where the program and the environment share some public variables while having also their own private variables~\cite{Reif84}.

Recently in~\cite{GS09,BGG09,BertrandGG17} two (mainly equivalent) models of concurrent stochastic games with imperfect information have been introduced. They permit to capture several known models (as those from~\cite{dAH00,chatterjeePHD,CDHR07}) while preserving the main decidability results. 

In this paper we consider the games as introduced in~\cite{GS09,BGG09,BertrandGG17} (we use the formalism of~\cite{GS09}). These are finite state games in which, at each round,
the two players choose concurrently an action and based on these
actions the successor state is chosen according to some fixed
probability distribution. {The resulting} infinite play is won by \Eve
if it satisfies a given \emph{objective}. The objectives we consider
here are reachability {(\ie a final state is eventually visited)},
safety {(\ie no forbidden state is visited)}, B\"uchi {(\ie some
final state is visited infinitely often)} and co-B\"uchi {(\ie no forbidden states is visited infinitely often)}.  Imperfect information is
modelled as follows: both players have an equivalence relation over
states and, instead of observing the exact state, they only observe
its equivalence class. {Intuitively, two equivalent states are 
indistinguishable by the corresponding player.}

In~\cite{GS09,BGG09,BertrandGG17} the authors were considering general strategies
where a player is allowed to use \emph{randomisation} when choosing
her/his next action. It was then shown, for B\"uchi objectives, that
one can decide whether \Eve has such a strategy $\phi$ that is
almost-surely winning against any strategy $\psi$ of \Adam (meaning
that an infinite play played according to $\phi$ and $\psi$ {is won
  by \Eve} with probability $1$). It was also established in
\cite{BGG09,BertrandGG17} that one can decide for co-B\"uchi objectives whether
\Eve has a positively winning strategy.

In the present work we restrict our attention to \emph{pure}
strategies, \emph{i.e.} we forbid the players to randomise when
choosing their actions.  
Our initial motivation for this work comes from automata theory. The
emptiness problem for automata on infinite trees can be described as
the problem of computing a winning strategy in a two-player game of
infinite duration. The required game model depends on the class of
automata that is considered.  In particular,~\cite{FPS13} proposes a
reduction of the emptiness problem for \emph{alternating} tree
automata to the existence of a pure winning strategy for \Eve in an
imperfect information game. For capturing the automaton model with a
qualitative acceptance condition as introduced in~\cite{CarayolHS14},
one furthermore needs stochastic games (and up to now this is the only
known method for checking emptiness of such automata). So one of our
aims is to obtain a toolbox and to understand the limits of this
method for checking emptiness of tree automata.

\oschanged{Another motivation for studying pure strategies comes from controller synthesis. Indeed, a classical way to consider an open system (\ie a controllable program interacting with an uncontrollable environment) is as a two-player game, and in this setting synthesising a controller for the program boils down to compute a winning strategy in the game for the player standing for the program. In the setting of open systems, imperfect information naturally arises, for instance when the program and its environment use private variables. A desirable property of a controller (in addition to its optimality) is its implementability which could be limited by its size but also by the required features. In the imperfect information setting, the main needed feature might be randomisation which is well-known to be non-trivial to implement in a non-biased fashion: hence, existence of an optimal non-randomised controller (equivalently a pure strategy) is a natural question.}

Our main results are the following.  On the negative side, by a
reduction of the value $1$ problem for probabilistic word automata
\cite{GO10}, we prove that even if \Adam is fully informed and \Eve is
totally blind (\emph{i.e.} all states are indistinguishable for her),
it is undecidable whether \Eve can positively win a safety game
(Section~\ref{sec:undecidability}).
Under the same restrictions, positive
 winning in B\"uchi games and almost-sure winning in co-B\"uchi games
 are proved to be undecidable by reduction from the emptiness problem
 for probabilistic $\omega$-word automata~\cite{BBG08}.

To obtain positive results, we have to impose restrictions on how \Adam is informed. We consider the case where he has perfect information and the case where he is more informed than \Eve\footnote{{We say that \Adam is more informed than \Eve when his equivalence relation on the states of the games refines that of \Eve. In particular, this is the case when \Adam is perfectly informed.}}. In both situations we show that it is decidable whether \Eve has a positively winning pure strategy in a reachability game (Section~\ref{sec:positive-reachability}). Using this
 result in a fixpoint computation, we prove that one can decide whether
 \Eve has an almost-surely winning pure strategy in a B\"uchi
 game (Section~\ref{sec:as}). Moreover, if it exists, such a strategy can be constructed and
 requires \emph{finite} memory. In both cases, we obtain matching upper and lower complexity bounds.

The decidability results for the special case where \Adam is perfectly informed were also obtained in~\cite{CD}. However, the technique we develop here is different and in particular uses the positive winning case as a toolbox, which later permits us to handle the more general case where \Adam is more informed than \Eve. And while~\cite{CD} focuses on reachability conditions and studies the memory required for winning strategies depending on how the players are informed, we focus on the case in which \Adam is better informed than \Eve (or even perfectly informed), and study different winning conditions. 

\reviewF{1}{Introduction: compare the results under general strategies vs. pure strategies}\answer{We addressed that in the next paragraph}

\oschanged{In our setting (restricting to pure strategies) the algorithmic complexity increases to 2-\exptime while it is only \exptime in the general setting (allowing randomised strategies) when \Adam is more informed than \Eve. Also note that in the latter setting one has decidability results in the case where no assumption is made on how \Adam and \Eve are informed while in our setting this question is left open.}

The resulting complete picture is summarised in the table at the end of this paper.


\section{Definitions}

A \defin{probability distribution} over a finite set $X$ is a mapping $d:X\rightarrow [0,1]$ such that $%
{\sum_{x\in X}d(x)=1}$. In the sequel we denote by $\distribution{X}$ the set of probability distributions over $X$.
Given some set $X$ and some equivalence relation $\sim$ over $X$, $[x]_{\sim}$ stands for the equivalence class of $x$ for $\sim$ and $X/_{\sim}=\{[x]_{\sim}\mid x\in X\}$ denotes the set of equivalence classes of $\sim$.
As usual we write $A^*$ (\resp $A^\omega$) for the set of finite (\resp infinite) words over some finite alphabet $A$. For $k\geq 0$ we denote by $A^{\geq k}$ (\resp $A^{\leq k}$) the set of words of length at least (\resp at most) $k$.

{
A \defin{concurrent arena with imperfect information} (or simply an \defin{arena}) is defined as a tuple \linebreak $\arena = (\states,\actions_\Evei,\actions_\Adami,\ftrans,\requiv_\Ei,\requiv_\Ai)$ where  $\states$ is a finite set of \defin{states};  $\actions_\Evei$ (\resp $\actions_\Adami$) is the (finite) set of \defin{actions} for \Eve (\resp \Adam);  $\ftrans: \states\times\actions_\Evei\times \actions_\Adami \rightarrow \distribution{\states}$ is the (total) transition function; and $\requiv_\Ei$ and $\requiv_\Ai$ are equivalence {relations} over states.}

A play in such an arena proceeds as follows. First it starts in some
initial state $s$. Then the first player, \Eve, picks an action
$\action_\Evei\in\actions_\Evei$ and, \emph{simultaneously} and
\emph{independently}, the second player, \Adam, chooses an action
$\action_\Adami\in\actions_\Adami$. Then a successor state is chosen
according to the probability distribution
$\ftrans(s,\action_\Evei,\action_\Adami)$, and the process restarts:
the players choose a new pair of actions that induces, together with
the current state, a new state and so on forever. Hence, a
\defin{play} is an infinite sequence $s_0(\action_\Ei^0,\action_\Ai^0)s_1(\action_\Ei^1,\action_\Ai^1)s_2\cdots$ in
$(\states\cdot(\actions_\Ei\times\actions_\Ai))^\omega$ such that for every $i\geq 0$, 
${\ftrans(s_i,\action_{\Evei}^{i},\action_{\Adami}^{i})(s_{i+1})>0}$. In the
sequel we refer to a prefix of a play ending by a state as a \defin{partial play}.

The intuitive meaning of $\requiv_\Ei$ (\resp $\requiv_\Ai$) is that two states $s_1$ and $s_2$ such that $s_1\requiv_{\Ei} s_2$ (\resp $s_1\requiv_{\Ai} s_2$) cannot be distinguished by \Eve (\resp by \Adam). We easily extend relation $\requiv_X$, with $X\in\{\Ei,\Ai\}$, to partial plays as follows. First, for any partial play $\play=s_0(\action_\Ei^0,\action_\Ai^0)s_1(\action_\Ei^1,\action_\Ai^1)\cdots s_k$ denote $[\play]_{\requiv_X}=[s_0]_{\requiv_X}[s_1]_{\requiv_X}\cdots [s_k]_{\requiv_X}$; then define $\play\requiv_X\play'$ if and only if $[\play]_{\requiv_X}=[\play']_{\requiv_X}$ .

{We say that \Adam is \defin{more informed} than \Eve if $\requiv_\Ai \subseteq \requiv_\Ei$, and \Adam is \defin{perfectly informed} if $\requiv_\Ai$ is the equality relation.}

{
\begin{example}
\label{ex:debut}
Consider the concurrent game with imperfect information depicted in Figure~\ref{fig:exemple-first}. 
 Let $\actions_\Evei = \actions_\Adami = \{a,b\}$. The initial state is $s_0$ and
 from $s_0$ if \Adam plays the action $a$  then any action played by \Eve leads with probability $\frac{1}{2}$ either to $s_1$ or to $s_2$. Similarly if \Adam plays $b$  then any action played by \Eve leads with probability $\frac{1}{2}$ either to $s_3$ or to $s_4$.  In the states $s_1, s_2, s_3$ and $s_4$, which are indistinguishable by \Eve, the action of \Adam has no impact. If \Eve plays $a$ from $s_1$ or $s_4$ or $b$ from $s_2$ or $s_3$ the play goes to the final state $f$ which is a sink state. Any other action by \Eve from one of those states leave the current state unchanged.
\begin{figure}
\begin{center}
\begin{tikzpicture}[transform shape,scale=1]
\begin{scope}
\tikzstyle{every node}=[minimum size=3mm,inner sep=0.5mm]
\tikzstyle{every loop}=[->,shorten >=1pt,looseness=7,]
\tikzstyle{loop top}=[in=55,out=125,loop]
\tikzstyle{loop down}=[in=-55,out=-125,loop]
\tikzstyle{loop right}=[in=35,out=-35,loop]
\tikzstyle{loop left}=[in=145,out=215,loop]

\node[draw,circle] (s0) at (0,1) {$s_0$};

\node[draw,circle] (s1) at (-2,-1) {$s_1$};
\node[draw,circle] (s2) at (-2,-2) {$s_2$};
\node[draw,circle] (s3) at (2,-1) {$s_3$};
\node[draw,circle] (s4) at (2,-2) {$s_4$};

\node[inner sep=0,minimum size=0mm] (in1) at (-0.5,-0.5) {};
\path[draw] (s0) to node[above left] {$ \any \barre a $ }(in1);
\path[->,draw,bend right] (in1) to  node[below][xshift=1mm] {$\frac{1}{2}$} (s1);
\path[->,draw,bend left] (in1) to node[right][xshift=1mm] {$\frac{1}{2}$} (s2);

\node[inner sep=0,minimum size=0mm] (in2) at (0.5,-0.5) {};
\path[draw] (s0) to node[above right] {$ \any \barre b $}(in2);
\path[->,draw,bend left] (in2) to node[below][xshift=-1mm] {$\frac{1}{2}$} (s3);
\path[->,draw,bend right] (in2) to node[left][xshift=-1mm] {$\frac{1}{2}$}(s4);

\node[draw,circle,double] (f) at (0,-4) {$f$} edge [in=30,out=-30,loop] ();
\node at (1.1,-4) {$\any\barre\any$};

\path[->,draw] (s1) to  node[above right][yshift=1mm][xshift=-1mm] {$a \barre \any$ } (f);
\path[->,draw] (s2) to  node[below left] {$b \barre \any$ } (f);

\path[->,draw] (s3) to  node[above left][yshift=1mm][xshift=1mm] {$b \barre \any$ } (f);
\path[->,draw] (s4) to  node[below right] {$a \barre \any$ } (f);

\draw[dashed] (-3.5,0) rectangle (3.5,-3);

\path (s1) edge [loop left] node[left] {$b\barre\any$} (s1);
\path (s2) edge [loop left] node[left] {$a\barre\any$} (s2);
\path (s3) edge [loop right] node[right] {$a \barre\any$} (s3);
\path (s4) edge [loop right] node[right] {$b\barre\any$} (s4);

\end{scope}
\end{tikzpicture}
\end{center}
\caption{A concurrent arena where \Adam is perfectly informed while \Eve cannot distinguish states $s_1,s_2,s_3$ and $s_4$.}\label{fig:exemple-first}
\end{figure}
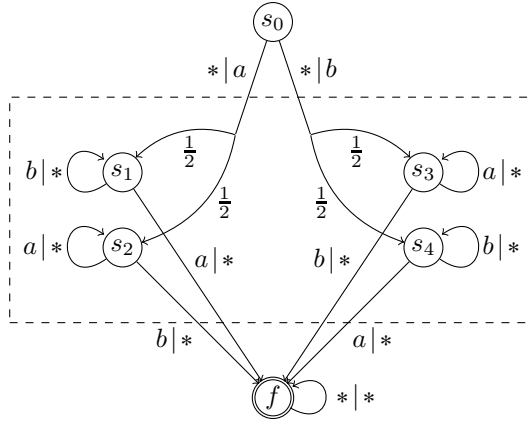

\end{example}
}

{In order to choose their moves the players respect strategies, and, for this, they may use all the information they have about what was played so far. However, if two partial plays are equivalent for $\requiv_\Ei$ (\resp $\requiv_\Ei$), then \Eve (\resp \Adam) cannot distinguish between them, and should behave the same. This leads to the following notion.
}

An \defin{observation-based pure strategy} (simply called a strategy in
the following) for \Eve is a function $\strat:
(\states/_{\requiv_\Ei})^*\rightarrow \actions_\Evei$, \emph{i.e.}, to choose
her next action, Eve considers the sequence of observations she {has
  seen} so far. We {overload} $\strat$ by writing $\strat(\play)$
instead of $\strat([\play]_{\requiv_\Ei})$: in
particular, a strategy $\strat$ {for \Eve} is such that
$\strat(\play)=\strat(\play')$ whenever $\play\requiv_\Ei \play'$ {(and
similarly for Adam)}.

A \defin{finite-memory strategy} for \Eve is a strategy that can be
computed by a finite automaton with output that reads the observation
sequence of the partial play and outputs the next action of Eve. We do
not give a precise technical definition because it is not needed in
this work. The size of such a strategy corresponds to the number of
states of the automaton.


{Strategies for \Adam are defined in a similar way by replacing $\requiv_\Ei$ by $\requiv_\Ai$.}

\reviewF{1}{Definitions: explicitly mention that even when Adam is not perfectly informed he observes the actions}
\answer{We added that in the case where Adam is more informed than Eve, as this is the only case where we have Adam imperfectly informed.}

{\begin{remark}\label{rk:action-visible}
 In our definition of a strategy we implicitly assume that the players only observe the
 sequence of states and not the corresponding sequence of
 actions. While the fact that a player does not observe what his adversary {has}
 played is reasonable (otherwise imperfect information on states would
 make less sense) one could object that the player should observe the actions
 she {has} played so far.
 However, as the players do not use randomisation in their strategies,
 they can always retrieve the actions they played so far.
 
\oschanged{Moreover, in the special case where \Adam is more informed than \Eve (as later studied in Section~\ref{subsection:AdamMoreInformed}), we can also note that, when playing against a fixed strategy of \Eve, he can always retrieve the actions she played so far (as he knows the strategy of \Eve and also $[\play]_{\requiv_\Ei}$ for any partial play $\play$).}
 \end{remark}}

Let $\arena = (\states,\actions_\Evei,\actions_\Adami,\ftrans,\requiv_\Ei,\requiv_\Ai,)$ be an arena, let $s_0\in \states$ be an initial state, $\strat_\Evei$ be a strategy for \Eve and $\strat_\Adami$ be a strategy for \Adam. 
First we let $\outcomes{s_0}{\strat_\Evei}{\strat_\Adami}$ to be the set
of all possible plays when the game starts in $s_0$ and when \Eve and
\Adam respectively follows $\strat_\Evei$ and $\strat_\Adami$: $\play = s_0(\action_\Ei^0,\action_\Ai^0)s_1(\action_\Ei^1,\action_\Ai^1)\cdots$ belongs to $\outcomes{s_0}{\strat_\Evei}{\strat_\Adami}$ iff  $$\ftrans(s_i,\strat_\Evei([s_0]_{/{\requiv_\Ei}}[s_0]_{/{\requiv_\Ei}}\cdots [s_i]_{/{\requiv_\Ei}}),\strat_\Adami([s_0]_{/{\requiv_\Ai}}[s_0]_{/{\requiv_\Ai}}\cdots [s_i]_{/{\requiv_\Ai}}))(s_{i+1})>0$$ for every $i\geq 0$.
Then we are interested in defining the probability of a (measurable) set of plays, knowing that \Eve (\resp \Adam) uses $\strat_\Evei$ (\resp $\strat_\Adami$).
This is done in the usual way (see \emph{e.g.}~\cite{chatterjeePHD}): once a pair $(\strat_\Evei,\strat_\Adami)$ of  strategies for both players is fixed, one is left with a (possibly infinite) Markov chain that naturally induces a probability space over the Borel $\sigma$-field generated by the cones, where for any partial play $\play$ starting in $s_0$ the cone for $\play$ is the set $\cone{\play}=\play\cdot ((\actions_\Ei\times\actions_\Ai)\cdot\states)^\omega$ of all infinite plays with prefix $\play$. We let $\proba{s_0}{\strat_\Evei}{\strat_\Adami}$ denote the corresponding probability measure over this space.

An \defin{objective} for \Eve is a (measurable) set $\objective$ of
plays: a play is won by \Eve if it belongs to $\objective$; otherwise
it is won by \Adam. A \defin{concurrent game with imperfect
  information} (simply called a game in the following) is a triple $\game = (\arena,s_0,\objective)$ where $\arena$ is an arena, $s_0$ is an initial state and $\objective$ is an objective. 
In the sequel we focus on the following special classes of {\hbox{$\omega$-regular}} objectives (note that all of them are Borel sets hence, measurable) that we define using a subset $F\subseteq \states$ of \defin{final} states.

A \defin{reachability objective} (\resp \defin{safety}) is of the form $(\states\cdot (\actions_\Ei\times\actions_\Ai))^*\fstates( (\actions_\Ei\times\actions_\Ai)\cdot\states)^\omega$ (\resp of the form $((\states \setminus F)\cdot (\actions_\Ei\times\actions_\Ai))^\omega$) : a play is winning if it contains (\resp does not contain) a final state. 

A \defin{B\"uchi objective} (\resp \defin{co-B\"uchi objective}) is of the form $\bigcap_{k\geq 0}(\states\cdot  (\actions_\Ei\times\actions_\Ai))^{\geq k}\fstates ((\actions_\Ei\times\actions_\Ai)\cdot\states)^\omega$ (\resp of the form $(\states\cdot  (\actions_\Ei\times\actions_\Ai))^* ((\states \setminus F)\cdot  (\actions_\Ei\times\actions_\Ai))^\omega$) : a play is winning if it goes infinitely often (\resp finitely often)  through  final states.

A reachability (\resp safety, B\"uchi, co-B\"uchi) game is a game equipped with a reachability (\resp safety, B\"uchi, co-B\"uchi) objective. In the sequel we may replace $\mathcal{O}$ by $F$ when it is clear from the context which objective we consider.

Fix a game $\game= (\arena,s_0,\objective)$.  { A strategy
  $\strat_\Evei$ for \Eve is \defin{surely winning} if, for any
  counter-strategy $\strat_\Adami$ for \Adam,
  $\outcomes{s_0}{\strat_\Evei}{\strat_\Adami}\subseteq \objective$. If
  such a strategy exists, we say that \Eve \defin{surely wins}
  $\game$.}  A strategy $\strat_\Evei$ for \Eve is \defin{almost-surely
  winning} (\resp \defin{positively winning}) if, for any
counter-strategy $\strat_\Adami$ for \Adam,
$\proba{s_0}{\strat_\Evei}{\strat_\Adami}(\objective)=1$ (\resp
$>0$). If such a strategy exists, we say that \Eve
\defin{almost-surely wins} (\resp \defin{positively wins}) $\game$.

In this paper, we are interested in deciding existence of
almost-surely/positively winning strategies for \Eve for
safety/reachability/B\"uchi/co-B\"uchi games.

\begin{example}\label{example:noDeterminacy} 
Consider the (perfect information) concurrent reachability game 
depicted below with $\win$
 as the unique final state. In state $\wait$, if both players choose
 the same action then they stay in state $\wait$ and otherwise they move to state $\win$. In state $\win$, all choices of actions stay in state $\win$. 
  \Eve does not have any almost-surely winning strategy.
  
  
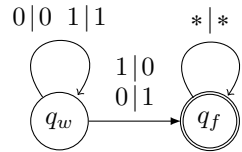
\begin{center}
  \begin{tikzpicture}[transform shape, scale=1]
\begin{scope}

\tikzstyle{every loop}=[->,shorten >=1pt,looseness=7,]
\tikzstyle{loop top}=[in=55,out=125,loop]
\node[draw,circle] (q_0) at (0,0) {$\wait$};

\node[draw,circle,double] (q_f) at (2,0) {$\win$};

\draw[->,>=latex] (q_0) to  node[above] {$\head\barre\tail$}
node[above][yshift=4mm] {$\tail\barre\head$ }(q_f);
\path (q_0) edge [loop top] node[above] {$\head\barre\head \;\; \tail\barre\tail$} (q_0);
\path (q_f) edge [loop top] node[above] {$\any\barre\any$} (q_f);

\end{scope}
\end{tikzpicture}
\end{center}

  Indeed, given any strategy
 $\strat_\Evei$ for \Eve, the counter-strategy $\strat_\Adami$ for \Adam mirroring the strategy of \Eve
 (\emph{i.e.}\  $\strat_\Adami=\strat_\Evei$) only allows for the play
   $\wait^\omega$ and hence, $\proba{s_0}{\strat_\Evei}{\strat_\Adami}(\objective)=0$.
    Similarly \Adam does not  have an almost-surely winning strategy.
 For any fixed strategy $\strat_\Adami$ of \Adam, any counter-strategy
 $\strat_\Evei$ for \Eve that satisfies $\strat_\Evei(\wait) \neq
 \strat_\Adami(\wait)$  is such that
 $\proba{s_0}{\strat_\Evei}{\strat_\Adami}(\objective)=1$.  
  \end{example}


\section{Undecidability Results} \label{sec:undecidability}

In this section we provide undecidability results for certain
combinations of types of winning strategies and objectives.  An easy
consequence of undecidability results for probabilistic
$\omega$-automata from~\cite{BBG08} is stated in the following theorem.
{In these reductions, \Eve plays alone and cannot distinguish any states of the game. The states and transitions of the game are those of the $\omega$-automaton and the strategy of \Eve corresponds to the input word.}

\begin{theorem}\label{Th:undecidabilityPOMDP}
{The decision problems whether \Eve almost-surely wins a given
co-B\"uchi game or positively wins a given B\"uchi game are undecidable
(even if the set of actions of \Adam is a singleton). }
\end{theorem}

\begin{proof}
Consider a probabilistic automaton $\mathcal{A}$ on $\omega$-words as
in~\cite{BBG08}. Now consider a concurrent game with imperfect
information where \Adam plays no role and where \Eve's actions are the
letters from the input alphabet $A$ of $\mathcal{A}$ and whose states
are the {ones} of the automaton. Moreover all states are
$\requiv_\Ei$-equivalent. Now the transition function of the game mimics
the {one} of the automaton. As \Eve does not observe anything, a (pure)
strategy $\phi$ of \Eve can be described as an infinite word $u_\phi$
in $A$ (the $i$-th letter being the $i$-th action played by \Eve), and
$\phi$ is almost-surely (\emph{resp.} positively) winning iff the
probability of a run of $\mathcal{A}$ over $u_\phi$ to be accepting is
$1$ (\emph{resp.} strictly positive). The undecidability results 
follow from the undecidability of the emptiness problem for
co-B\"uchi (\emph{resp.} B\"uchi) probabilistic automaton with the
almost-sure (\emph{resp.} positive) semantics~\cite{BBG08}.  
\end{proof}

In the following we prove that the existence of positively winning
strategies for safety objectives is undecidable.  Our result is based
on the undecidability of the value 1 problem for probabilistic
automata on finite words~\cite{GO10}. For simpler use in our reduction
we reformulate this problem in terms of games.

Consider the class of concurrent reachability games $\game$ with
imperfect information with the following properties. \Eve is blind
(\emph{i.e.} $\requiv_\Ei$ consists of a unique equivalence class),
and \Adam has no impact on the game (\emph{i.e.} his set of actions is
a singleton). Furthermore, there is a special action $\sharp$ that
\Eve can play at any time, and that leads (depending on the current
state) either to a final sink state or to a non-final sink state. The
final sink state is the only final state. Intuitively, one can think
of such a game as one where \Eve plays a sequence of actions and then
declares by $\sharp$ that she stops (and she wins if she stopped in a
state that leads to the winning sink).

{We refer to this type of game as \emph{probabilistic automaton game} (PA game)
because it corresponds to probabilistic automaton on finite words (see
\cite{Paz71} for an introduction to probabilistic automata): a
strategy of \Eve corresponds to a finite word followed by $\sharp$
(without playing $\sharp$ \Eve surely loses), and the probability that
it is winning is the probability of the word to be accepted in the
automaton. Then we have the following result, which directly follows
from the undecidability of the value 1 problem for probabilistic automata~\cite{GO10}.

\begin{lemma}\label{lem:value1}
For a given a PA game, it is undecidable whether
Eve has for each $0 < \varepsilon < 1$ a strategy that is winning with
probability $(1-\varepsilon) < p < 1$.
\end{lemma} 
}

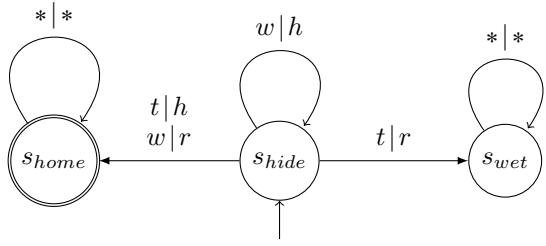
\begin{figure}[t]
\centering
\begin{tikzpicture}[transform shape, scale=1]
\begin{scope}

\tikzstyle{every loop}=[->,shorten >=1pt,looseness=7,]
\tikzstyle{loop top}=[in=55,out=125,loop]
\node[draw,circle,initial,initial text=,initial where = below] (qhide) at (0,0) {$s_{hide}$};
\node[draw,circle,double] (qhome) at (-3,0) {$s_{home}$};
\node[draw,circle] (qwet) at (3,0) {$s_{wet}$};

\draw[->,>=latex] (qhide) to  node[above] {$w\barre r$} node[above][yshift=4mm] {$t\barre h$ }(qhome);
\draw[->,>=latex] (qhide) to  node[above] {$t\barre r$} (qwet);
\path (qhide) edge [loop top] node[above] {$w\barre h$} (qhide);
\path (qhome) edge [loop top] node[above] {$\any\barre\any$} (qhome);
\path (qwet) edge [loop top] node[above] {$\any\barre\any$} (qwet);
\end{scope}
\end{tikzpicture}
\caption{The Hide-or-Run game}\label{figure:HideOrRun}
\end{figure}

Our reduction that uses Lemma~\ref{lem:value1} starts from an
example of a concurrent safety game $\game_{HR}$ known as \emph{Hide-or-Run}
\cite{dAHK07} (see Figure~\ref{figure:HideOrRun}). In this game,
{\Adam can choose between hiding ($h$) and running ($r$), and \Eve} can choose between waiting ($w$) and  throwing ($t$) her {only} snowball. 
 If \Adam hides and \Eve waits, the game stays in state $s_{hide}$. If \Adam runs and \Eve throws the snowball, then \Adam is hit, and the game proceeds to sink state $s_{wet}$. In all other cases, \Adam gets home (either he runs without being hit or he can safely run after \Eve has thrown her snowball) and the game proceeds to sink state $s_{home}$. This is a safety game where \Eve wants to avoid visiting $s_{home}$.

In~\cite{dAHK07} it is shown that \Eve can only win by using a
randomised strategy that plays action $w$ in round $i$ with
probability $p_i$ such that $0<p_i<1$ for every $i$ and $\prod_i p_i>
0$ (for this, \Eve does not have to distinguish the states). 

Now the idea is to incorporate a gadget in $\game_{HR}$ that permits
\Eve to simulate random choices while playing deterministically.

\begin{theorem}
\label{theo:undecidable-positive-safety}
It is undecidable whether \Eve positively wins in a safety game (\resp co-B\"uchi game), even if $\requiv_\Ei$ consists of a single equivalence class. 
\end{theorem}

\begin{proof}
Consider a probabilistic automaton game $\game$ with a set of actions
disjoint from the one in the game $\game_{HR}$.  Let $\game_r$ and
$\game_h$ be two disjoint copies of $\game$ where we removed the two
states reachable by \Eve playing $\sharp$ (the $\sharp$-edges are
redirected as described below).

In the game $\game_{HR}'$ (see Figure~\ref{Figure:HRB}), the
concurrent choices of the actions in $\game_{HR}$ are simulated by the
imperfect information. All states are indistinguishable by Eve. First
\Adam makes his choice $r$ or $h$ from $s_{hide}$ (\Eve's action has
no impact). The game then moves to the initial state of $\game_r$ or
$\game_h$, depending on the choice of Adam (ignore the action $cheat$
for the moment, which is explained later). Because of the imperfect
information \Eve does not observe Adam's choice.

In $\game_r$ and $\game_h$ we removed the target states of $\sharp$
but \Eve can still play $\sharp$: if in $\game$ it was leading to the
final state it now behaves as \Eve playing $w$ from $s_{hide}$, and
otherwise it behaves as \Eve playing $t$ from $s_{hide}$ (see Figure
\ref{Figure:HRB}). 

Finally, in order to prevent \Eve from playing an {infinite sequence of
actions without $\sharp$}, we add an extra small gadget where \Adam is
allowed to declare that \Eve will cheat.  If he plays $cheat$ from
$s_{hide}$ this leads to a new state $s_c$ where the following may
happen depending on the next move of \Eve (the action of \Adam has no
impact): if she plays $\sharp$ from $s_c$ then the play goes to a
sink state $s_w$ (that is not final); if she does not play $\sharp$
from $s_c$ then with probability $1/2$ the play stays in $s_c$ and
with probability $1/2$ the play goes to a sink final state
$s_l$. Hence, from $s_c$ if she never plays $\sharp$, then the play
almost-surely ends in $s_l$.

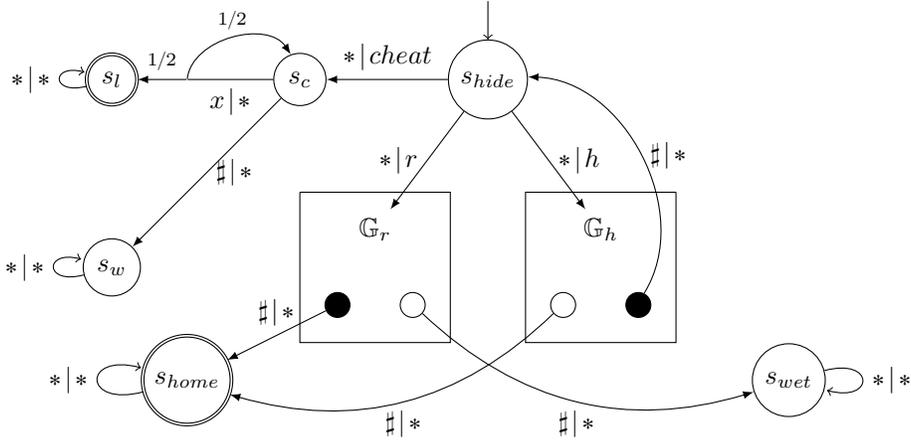
\begin{figure}[t]
\begin{center}
\begin{tikzpicture}[transform shape, scale=1]
\begin{scope}

\tikzstyle{every loop}=[->,shorten >=1pt,looseness=7,]
\tikzstyle{loop top}=[in=55,out=125,loop]
\node[draw,circle,initial,initial text=,initial where = above] (qhide) at (0,4) {$s_{hide}$};
\node[draw,circle,double] (qhome) at (-4,0) {$s_{home}$};
\node[draw,circle] (qwet) at (4,0) {$s_{wet}$};
\path (qhome) edge [loop left] node[left] {$\any\barre \any$} (qhome);
\path (qwet) edge [loop right] node[right] {$\any\barre \any$} (qwet);

\node (nameGr) at (-1.5,2) {$\game_r$};
\node[draw,circle,fill] (finalGr) at (-2,1) {};
\node[draw,circle] (nonfinalGr) at (-1,1) {};
\draw (-2.5,2.5) -- (-2.5,0.5) -- (-0.5,0.5) -- (-0.5,2.5) -- (-2.5,2.5);

\node (nameGw) at (1.5,2) {$\game_h$};
\node[draw,circle,fill] (finalGw) at (2,1) {};
\node[draw,circle] (nonfinalGw) at (1,1) {};
\draw (2.5,2.5) -- (2.5,0.5) -- (0.5,0.5) -- (0.5,2.5) -- (2.5,2.5);

\node[draw,circle] (sc) at (-2.5,4) {$s_{c}$};
\node[draw,circle] (sw) at (-5,1.5) {$s_{w}$};
\node[draw,circle,double] (sl) at (-5,4) {$s_{l}$};
\node[inner sep=0,minimum size=0mm] (inter) at (-4,4) {};
\draw[->,>=latex] (qhide) to  node[above] {$\any\barre cheat$} (sc);
\draw[->,>=latex] (sc) to node[right,pos=0.5] {$\sharp\barre \any$} (sw);
\path (sw) edge [loop left] node[left] {$\any\barre \any$} (sw);
\path (sl) edge [loop left] node[left] {$\any\barre \any$} (sl);
\draw[-] (sc) to node[below,pos=0.5] {$x\barre \any$} (inter);
\draw[->,>=latex] (inter) to node[above,pos=0.5] {\scriptsize $1/2$} (sl);
\draw[->,>=latex] (inter) to [bend left=70] node[above,pos=0.5] {\scriptsize $1/2$} (sc);

\draw[->,>=latex] (qhide) to  node[left] {$\any\barre r$} (nameGr);
\draw[->,>=latex] (finalGr) to node[above] {$\sharp\barre \any$} (qhome);
\draw[->,>=latex] (nonfinalGr) to[bend right] node[below] {$\sharp\barre \any$} (qwet);

\draw[->,>=latex] (qhide) to  node[right] {$\any\barre h$} (nameGw);
\draw[->,>=latex] (finalGw) to[bend right=60] node[right] {$\sharp\barre \any$} (qhide);
\draw[->,>=latex] (nonfinalGw) to[bend left] node[below] {$\sharp\barre \any$} (qhome);

\end{scope}
\end{tikzpicture}
\caption{The modified version of \emph{Hide-or-Run}:
  $\game_{HR}'$. Black states in $\game_r$/$\game_h$ correspond to
  states from which $\sharp$ led to the final state in $\game$, and
  $x$ denotes any letter different from $\sharp$.}\label{Figure:HRB}
\end{center}
\vspace{-0.5cm}
\end{figure}

\reviewF{1}{iff Eve in G has an almost sure winningÃ¢ÂÂ¦ you are not proving this, but rather iff for every
 $\epsilon>0$ there exists a strategy such thatÃ¢ÂÂ¦}
\answer{You are right. We changed this.}

Let $\game_{HR}'$ be this new game, where we recall that all states
are indistinguishable for \Eve, $s_{hide}$ is the initial state and
$\{s_{home},s_l\}$ are the final states.  \oschanged{We claim that \Eve
positively wins game $\game_{HR}'$ iff \Eve in $\game$ has strategies winning with probability arbitrarily close to $1$}. Indeed,
consider a strategy $\phi$ for \Eve in $\game_{HR}'$. As \Eve cannot
distinguish any state in $\game_{HR}'$, and does not observe the
actions played by Adam, $\phi$ is independent of Adam's choices. 

{If the strategy of {\Eve} consists in playing $\sharp$ only finitely often, it cannot be positively
winning as it suffices for \Adam to wait for the last $\sharp$ and
then play $cheat$. More precisely, the strategy  of {\Adam} consists
in playing (in state $s_{\textrm{hide}}$) the action $h$ whenever
\Eve's strategy  will still play $\sharp$ in the future, and $cheat$
if \Eve will never play $\sharp$ in the future. It can be shown that following this strategy \Adam wins against the strategy of \Eve with probability $1$.}

Thus, in the following we only consider strategies $\phi$ of {\Eve}
that play $\sharp$ infinitely often. An equivalent description of such
strategies $\phi$ is by a sequence $(\phi_i)_{i\geq 1}$ of strategies
for \Eve in $\game$: $\phi$ consists in playing an arbitrary letter
then playing as $\phi_1$ until playing $\sharp$, then playing an
arbitrary letter, then playing as $\phi_2$ until playing $\sharp$ and
so on (the arbitrary letter is used here when \Adam chooses to move to
$\game_r$, $\game_h$ or $s_c$).  

For one direction, assume that $\varphi$ is positively winning in
$\game_{HR}'$. Let $p_i$ be the probability that \Eve wins in $\game$
when playing according to $\phi_i$. Then, from the properties of
$\game_{HR}$, it follows that $\phi$ is winning iff $0<p_i<1$ for every
$i\geq 1$ and $\prod_i p_i> 0$. This implies that the sequence
$(p_i)_{i\geq 1}$ converges to $1$ and hence the $\varphi_i$ are
strategies as in Lemma~\ref{lem:value1}.

Conversely, if \Eve has strategies winning with probabilities
arbitrarily close to $1$ as in Lemma~\ref{lem:value1}, then one can
choose the $\phi_i$ such that {$1 > p_i \ge
  1-\dfrac{1}{(i+1)^{2}}$} which ensures $0<p_i<1$ for every $i\geq 1$
and $\prod_i p_i> 0$. Indeed, $$\prod_{i \geq 1} 1 - \frac{1}{(i+1)^{2}}=\lim_{m \rightarrow \infty} \prod_{i=1}^{m} 1-\frac{1}{(i+1)^{2}} = \lim_{m \rightarrow \infty}  \frac{m+2}{2m+2} = \frac{1}{2}$$ This family $\varphi_i$ defines a strategy for
\Eve in $\game_{HR}'$. Again using the properties of $\game_{HR}$,
this implies that \Eve positively wins against all strategies of Adam:
either no outcome ever reaches $s_c$, in which case $\game_{HR}$ is
simulated, or if an outcome reaches $s_c$, then it does with positive
probability, and then it also reaches $s_w$ with positive probability.
\end{proof}


\section{Positive Winning in Reachability Games} \label{sec:positive-reachability}

We now address the decidability of whether \Eve positively wins in a reachability game, and
we show decidability (and matching lower bounds) for the case where 
\begin{inparaenum}[(i)]
\item \Adam is perfectly informed and 
\item \Adam is more informed than \Eve.
\end{inparaenum}

For the rest of this section fix an arena $\arena = (\states,\actions_\Evei,\actions_\Adami,\ftrans,\requiv_\Ei,\requiv_\Ai)$ and a set of final states $F\subseteq S$.

To later address almost-sure winning (Section~\ref{sec:as}) we need to consider games that
may start in different states, and
we are interested in strategies that are winning from all of these
states. 
For this reason, we define for any
subset $B$ of states a game $(\arena,B,\objective)$ that is played as
follows: there is a new initial step where \Adam picks a state $s_0$
in $B$ and then the play {proceeds as} in
$(\arena,s_0,\objective)$. Hence, a strategy $\strat$ for \Eve in such a
game is almost-surely (\emph{resp.} positively) winning iff $\strat$ is
almost-surely (\emph{resp.} positively) winning in
$(\arena,s_0,\objective)$ for \emph{every} state $s_0\in B$.

\subsection{Winning in a Finite Number of Moves.}

We start with a general result that does not depend on how the players are informed.
It states that if \Eve can positively win in a reachability game then she can do so within a bounded number of moves.

\begin{proposition}\label{proposition:positiveWinningReachability}
Let $B\subseteq S$ be a subset of states and assume that \Eve has a positively winning strategy $\strat$ in the reachability game $(\arena,B,F)$. Then, there is a bound $N$ and some $0<\epsilon_{{B}}\leq 1$ such that whenever \Eve respects $\strat$ in {the} game $(\arena,B,F)$, the probability that the resulting play visits a final state within the $N$ first moves is at least $\epsilon_B$. 
\end{proposition}

\begin{proof}
For any $N>0$, any $s\in B$ and any strategy $\psi_N$ for Adam, 
call $p_N^{\psi_N,s}$ the probability of the event \emph{"a play in $(\arena,s,F)$, where \Eve respects $\strat$ and \Adam respects $\psi_N$ visits a final state within the $N$ first moves"}. 

Let $x_N^{\psi_N}= \min\{p_N^{\psi_N,s}\mid s\in B\}$. We aim to show that there exists some $N>0$ such that for each strategy $\psi_N$ for Adam, $x_N^{\psi_N}>0$.

For this, we reason by contradiction, assuming that for any bound $N>0$, \Adam has a counter strategy $\psi_N$ such that $x_N^{\psi_N}=0$. In particular, there is a state $s\in B$ such that $p_N^{\psi_N,s}=0$ for infinitely many $N$. {Hence, we can assume that the $\psi_N$ are such that} $p_N^{\psi_N,s}=0$ for every $N\geq 0$ (as to get the property for some $N$ \Adam can always use the strategy for some $N'>N$).

\reviewF{1}{fourth / fifth paragraph of the proof of the Proposition 4.1: reword these two sentences}\answer{Reworded.}

Using $(\psi_N)_{N\geq 0}$ we define a strategy $\psi$ for \Adam as
follows. We first let $I_0=\mathbb{N}$ be the set of naturals. Next we
define $\psi$ and $(I_k)_{k\geq 0}$, a decreasing sequence (for
inclusion) of \emph{infinite} subsets of the naturals. First we sort
partial plays by increasing length. We assume that $\psi$ is defined
on all partial plays of length smaller than $k$ (hence initialization
for $k=0$ comes for free) and for plays of length $k+1$ we do the
following. \oschanged{As there are finitely many plays of length $k+1$ while $I_k$
is infinite there exists an infinite subset $I_{k+1}\subseteq I_k$
such that, for all $j_1,j_2\in I_{k+1}$, both strategies $\psi_{j_1}$ and $\psi_{j_2}$ agree on plays of
length $k+1$; we define $\psi$ to behave accordingly on plays of length $k+1$.

Then, the following is a direct consequence of the definitions of $\psi$ and $(I_k)_{k\geq 0}$:
for every $k\geq 0$, the set $I_k$ is infinite; and for every $j\in I_k$ and every partial play $\play$ of length smaller than $k$, both $\psi$ and $\psi_j$ agree on $\play$.
}

In particular it implies that $x_N^{\psi}=0$ for every $N\geq 0$: indeed, $x_N^{\psi_{M}}=0$ for any $M\geq N$ and $\psi$ agrees with all $\psi_{M}$ with $M\in I_N$ (and as $I_N$ is infinite such an $M$ exists). 
Finally, as $0\leq \proba{s}{\strat}{\psi}(\objective)\leq\sum_{N\geq 0} x_N^{\psi}=0$ (here $\objective$ denotes the reachability objective defined by $F$), we conclude that $\proba{s}{\strat}{\psi}(\objective)=0$ which contradicts our initial assumption of $\strat$ being positively winning in $(\arena,s,F)$. 

The fact that there is some $\epsilon_B>0$ such that $\strat$ ensures to
reach a final state in less than $N$ moves with a probability greater
than $\epsilon_B$ is a direct consequence of the fact that one bounds the number of moves by $N$.
\end{proof}

\begin{remark}\label{rk:PosReachNoProba}
Proposition~\ref{proposition:positiveWinningReachability} implies that finite memory suffices for \Eve to positively win in a reachability game. Indeed, it suffices to follow $\strat$ for the $N$ first moves and then play the same action forever.

Another important consequence is that the values of the probabilities do not have any influence on whether \Eve positively wins in a reachability game. More precisely consider another arena $\arena'$ that is exactly as $\arena$ except that its transition function $\delta'$ is  such that for every state $s$ and every pair of actions $(\action_\Ei,\action_\Ai)$ one has $\delta(s,\action_\Ei,\action_\Ai)= 0$ iff $\delta'(s,\action_\Ei,\action_\Ai)= 0$. Then \Eve positively wins in the reachability game $(\arena,B,F)$ iff she positively wins in the reachability game $(\arena',B,F)$.
\end{remark}

\subsection{Positively Winning When \Adam Is Perfectly Informed}

We now assume that \Adam is perfectly informed.

Consider for every $n \geq 0$, the objective $\ReachN(\fstates)=(\states\cdot(\actions_\Ei\times\actions_\Ai))^{< n}\fstates ((\actions_\Ei\times\actions_\Ai)\cdot\states)^{\omega}$ where a final state has to be visited within the first $n$ steps. 
The following inductively characterises the sets $B$ for which \Eve can win $(\arena,B,\ReachN(\fstates))$.

\begin{proposition}
\label{prop:reachN-ind}
Let $B \subseteq S$ be a set of pairwise $\requiv_\Ei$-equivalent states and let $n > 0$. \Eve positively wins $(\arena,B,\ReachN(\fstates))$ if and only if there exists an action $\action_{\Ei} \in \actions_{\Ei}$ and a set $B' \subseteq S$ such that
\begin{itemize}
\item \Eve positively wins $(\arena,B',\ReachNP(\fstates))$,
\item for every $s \in B \setminus F$ and for every $\action_{\Ai} \in \actions_{\Ai}$, there exists $s' \in B'$ such that $\delta(s,\action_{\Ei},\action_{\Ai})(s')>0$.
\end{itemize}
\end{proposition}

\begin{proof}
\reviewF{1}{last line of the first paragraph of the proof of Proposition 4.3: winning in $(\arena,s',\ReachN(\fstates))$, $n$ should be $n-1$}\answer{Fixed}
Fix a set $B \subseteq S$ of pairwise equivalent states and an integer $n >0$.
For the direct implication assume that \Eve has  a positively winning strategy $\strat$ in $(\arena,B,\ReachN(\fstates))$. Let $\action_{\Ei} = \strat([B]_{\requiv_\Ei})$ be the first action played by \Eve and let $\strat'$ be the strategy followed by \Eve after this first step (\ie $\strat'(\play) = \strat([B]_{\requiv_\Ei} \cdot \play)$ for every partial play $\play$). Let $B'$ be the set of states $s'$ such that $\strat'$ is positively winning in \oschanged{$(\arena,s',\ReachNP(\fstates))$}. 

\reviewF{1}{second paragraph of the proof of Proposition 4.3: and some state $s_0\in B$ such that for every $s\in S$ (should
 be $s\in B'$), remove the second "such that" $\delta(s_0,\sigma_E,\sigma_A, s)>0$ (should be equal to 0).}\answer{Fixed!}
We claim $\action_{\Ei}$ and $B'$ satisfy the property of the statement. First and by definition $\strat'$ is positively winning in $(\arena,B',\ReachNP(\fstates))$. \oschanged{For the second property assume toward a contradiction that there exist some $\action_{\Ai} \in \actions_{\Ai}$ and some state $s_{0} \in B$ such that for every $s \in B'$, $\delta(s_{0},\action_{\Ei},\action_{\Ai},s)=0$ then $\strat'$ is not positively winning in $(\arena,s,\ReachNP(\fstates))$ (\ie there exists a strategy $\psi_{s}$ of \Adam  such that $\proba{s}{\strat'}{\psi_{s}}(\ReachNP(\fstates)) = 0$).} Consider the strategy $\psi$ for \Adam consisting in playing first $\sigma_{\Ai}$ and then the $\psi_{s}$ corresponding to the observed state $s$
(\ie $\psi(s_{0})=\sigma_{A}$ and $\psi(s_{0}(\action_\Ei,\action_\Ai)\play)=\psi_{s}(\play)$
for any partial play $\play$ starting with $s \in B'$). We have  the following contradiction:
\[
\proba{s_0}{\strat}{\psi}(\ReachN(\fstates)) = \sum_{s \in B'} \delta(s_{0},\action_{\Ei},\action_{\Ai})(s) \cdot \underbrace{\proba{s}{\strat'}{\psi_{s}}(\ReachNP(\fstates))}_{=0} = 0.  
\]

For the converse implication assume that there exists an action $\action_{\Ei} \in \actions_{\Ei}$ and a set $B' \subseteq S$ satisfying the properties of the statement. Fix a positively winning strategy $\strat'$ for \Eve in $(\arena,B',\Reach^{\leq n-1}\!(\fstates))$.

Consider the strategy $\strat$ for \Eve in $(\arena,B,\ReachN(\fstates))$ consisting of first playing $\action_{\Ei}$ and then following $\strat'$ (\ie $\strat([B]_{\requiv_\Ei})=\action_{\Ei}$ and $\strat([B]_{\requiv_\Ei}\cdot [\play]_{\requiv_\Ei})=\strat'([\play]_{\requiv_\Ei})$ for every partial play $\play$). We claim that this strategy is positively winning in $(\arena,B,\ReachN(\fstates))$.

Indeed, let $\psi$ be a strategy for \Adam in $(\arena,B,\ReachN(\fstates))$ and let $s_{0} \in B$. Let $\sigma_{\Ai}$ be the first action played by \Adam when using $\psi$ and let $\psi'$ be the strategy followed by \Adam after this first step (\ie $\psi'(\play) = \psi(s_0(\action_\Ei,\action_\Ai) \cdot \play)$ for every partial play $\play$). By definition of $B'$, there exists $s' \in B'$ such that $\delta(s_{0},\action_{\Ei},\action_{\Ai})(s')>0$. We have:
 \[
\proba{s_0}{\strat}{\psi}(\ReachN(\fstates)) \geq  \delta(s_{0},\action_{\Ei},\action_{\Ai})(s') \cdot \proba{s'}{\strat'}{\psi'}(\ReachNP(\fstates)) >0.
\]
\end{proof}

Now, consider the increasing family of sets $(\mathcal{W}_{i})_{i\geq 0}$ defined by:  
\begin{itemize}
\item $\mathcal{W}_{0} = \{ B  \mid B \subseteq F \}$
\item $\mathcal{W}_{i+1}= \{ B \subseteq S \mid \forall r \in B,\ \exists \action_\Ei\  \exists B'\in \mathcal{W}_i \  \text{ s.t. } \forall s\requiv_\Ei r \text{ with } s {\in B \setminus F},\ \forall \action_\Ai\ \exists s'\in B' \text{ s.t. } \delta(s,\action_\Ei,\action_\Ai)(s')>0\}$
\end{itemize}
and call $\mathcal{W}$ its limit. Then the following is a simple consequence of Proposition~\ref{prop:reachN-ind}.

\begin{theorem}\label{theo:semiGame}
Let $B\subseteq\states$ {be a non-empty set}. \Eve has a positively winning strategy $\strat$ in the game $(\arena,B,\fstates)$ if and only if  $B\in\mathcal{W}$. In particular it can be decided in time exponential in $|S|$ whether \Eve has a positively winning strategy. {If such a strategy exists, one can construct one that uses the set $2^\states$ as memory, and this strategy guarantees to positively reach a final state within the $2^{|\states|}$ first moves.}
\end{theorem}

\begin{proof}
Using Proposition~\ref{prop:reachN-ind}, by a direct induction on $n$ one gets that $B\neq \emptyset$ belongs to $\mathcal{W}_{n}$ if and only if \Eve positively wins $(\arena,B,\ReachN(\fstates))$. 

For any $B \in \mathcal{W}$, we denote by $\mathrm{rk}(B)$ the smallest $n$ such that $B \in \mathcal{W}_n$. Now, for any $B \in \mathcal{W}$, we define a strategy for \Eve denoted $\phi_{B}$ that uses $\mathcal{W}$ as a finite memory. Initially the memory is $B$. For a partial play $\play$ ending in a state in some equivalence class $[s]_{\requiv_{\Ei}}$ and assuming that the memory is $B'$, we define the strategy as follows:
\begin{itemize}
\item If $\mathrm{rk}(B')>0$ and if there exists \oschanged{$r \in [s]_{\sim_{\Ei}} \cap B'$}, {then by definition of $(\mathcal{W}_i)_{i\geq 0}$ there exists some action $\action_{\Ei}$ and some set $B''$ such that the following holds:}\reviewF{1}{$r \in [s]_{\sim_{\Ei}} \cap B$ (should be $B'$)}\answer{fixed}
\begin{itemize}
\item $\mathrm{rk}(B'')=\mathrm{rk}(B')-1$,
\item $\forall r'\requiv_\Ei r \text{ with } r' \in B \setminus F, \forall \action_\Ai,\ \exists s'\in B'', \text{ s.t. } \delta(s,\action_\Ei,\action_\Ai)(s')>0$.
\end{itemize}
{Then we let $\phi_{B}(\play)=\action_{\Ei}$ and update the memory to $B''$.}
\item In all other cases, we take $\phi_{B}(\play)$ to be an arbitrary action and 
update the memory to $\emptyset$.
\end{itemize}

By induction on $n$, we show that for every non-empty $B \in \mathcal{W}_{n}$,
the strategy $\varphi_{B}$ is positively winning in $(\arena,B,\ReachN(\fstates))$. The base case is immediate. Assume that the property is established for $n-1\geq0$.
Let $B$ be a non-empty element of {$\mathcal{W}_{n}$}. Let $s_{0} \in B$, $\sigma_{\Ei}=\phi_{B}([s_{0}]_{\requiv_{\Ei}})$ and $B' \in \mathcal{W}_{n-1}$ be 
the memory of $\phi_{B}$ after the first move.
Let $\psi$ be a strategy for \Adam. Let $\sigma_{\Ai}$ be the first action played by \Adam when using $\psi$ and let $\psi'$ be the strategy followed by \Adam after this first step, \ie $\psi'(\play) = \psi(s_0 \cdot {(\action_\Ei,\action_\Ai)\cdot}\play)$ for every partial play $\play$. By definition of $B'$, there exists $s' \in B'$ such that $\delta(s_{0},\action_{\Ei},\action_{\Ai})(s')>0$. Hence, we have:
 \[
\proba{s_0}{\phi_{B}}{\psi}(\ReachN(\fstates)) \geq  \delta(s_{0},\action_{\Ei},\action_{\Ai})(s') \cdot \proba{s'}{\phi_{B'}}{\psi'}(\ReachNP(\fstates)) >0.
\]
{which concludes the proof.}
\end{proof}

The following is a restatement of the end  of Theorem~\ref{theo:semiGame}.
\begin{corollary}\label{cor:boundingN}
In Proposition \ref{proposition:positiveWinningReachability}, when \Adam is perfectly informed, one can always choose $\strat$ such that $N\leq 2^{|\states|}$.
\end{corollary}
%


\subsection{Automaton-Compatible Strategies}\label{sec:automaton-driven}

The aim of this section is to refine Theorem~\ref{theo:semiGame} to positively winning strategies that satisfy further constraints. The motivation is that in Section~\ref{sec:as} we compute almost-sure winning strategies for B\"uchi conditions using a fixpoint computation. In one iteration of this computation, we compute positively winning strategies for reachability that satisfy an extra constraint (roughly, that \Eve can positively win the reachability game while ensuring that she can win another round of the reachability game once the target set is reached). This further constraint is expressible by finite automata that read partial plays and restrict the set of admissible next actions for \Eve. Thus, below we develop the notion of a strategy that is compatible with such an automaton and then later apply it to the specific setting that we need.

Let $\trans = ( Q,\actions_\Ei\times\states_{/_{\requiv_\Ei}},q_0,q_s,\deltaT,\act) $ be a deterministic finite automaton with input alphabet $\actions_\Ei\times\states_{/_{\requiv_\Ei}}$, a finite set of states $Q$, an initial state $q_0$, a sink state $q_s$, a transition function $\deltaT:Q\times (\actions_\Ei\times\states_{/_{\requiv_\Ei}})\rightarrow Q$  and a function $\act:Q\rightarrow 2^{\actions_\Ei}$ associating with any state of $\trans$ a subset of actions for \Eve. Moreover, we require that the following holds
\begin{itemize}
\item  $\act(q)=\emptyset$ if and only if $q=q_s$. 
\item For every state $q$ and for every $(\action,x)\in \actions_\Ei\times\states_{/_{\requiv_\Ei}}$ one has $\deltaT(q,(\action,x))=q_s$ if and only if $\action\notin\act(q)$.
\end{itemize}

Such a machine associates with any partial play $\play$ a unique state $q_\play$ defined by $q_{s_0}=q_0$ and $q_{\play\cdot (\action_\Ei,\action_\Ai)\cdot s}= \deltaT(q_\play,(\action_\Ei,[s]_{\requiv_\Ei}))$; 
it also permits to associate with any partial play a subset of actions by letting $\actT(\play) = \act(q_\play)$. 

A strategy $\strat$ of \Eve is \defin{$\trans$-compatible} if for every partial play $\play$ where \Eve respects $\strat$ one has $\strat(\play)\in\actT(\play)$. Note that it implies that $q_\play\neq q_s$.

\reviewF{1}{the intention of this remark is clear but the remark itself is not, e.g. it seems like it's
 assuming that either the sink state is reached or all transitions from the initial state go back to it}\answer{We rephrased the remark below}
 

\oschanged{ \begin{remark}\label{rk:transTrivial}

 Consider the special case of the automaton $\trans_0$ defined as follows: $Q$ consists only of two states, the initial state and the sink state; $\deltaT(q_0,(\action,x))=q_0$ and $\deltaT(q_s,(\action,x))=q_s$ for any $(\action,x)\in\actions_\Ei\times\states_{/_{\requiv_\Ei}}$ (\ie all transitions are looping); and $\act$ equals all actions $\actions_\Ei$ in the initial state. Then it follows that any strategy is $\trans_0$-compatible. 
 
 Hence, by considering the special case of $\trans_0$, any result we obtain later will also hold if we drop the $\trans$-compatibility constraint.
 \end{remark}
}

In Section~\ref{sec:as} and for the proof of Theorem~\ref{theo:posTcom}, we work with automata that compute the belief of \Eve along a play, as explained below.
{For an initial belief set $B_0 \subseteq S$ of pairwise $\requiv_\Ei$-equivalent states,} the belief (also known as knowledge) $\belief{B_0}{\Ei}(\lambda)$ of \Eve after a partial play $\lambda$ starting in a state of $B_{0}$, intuitively corresponds to the set of possible states that can have been reached in a play $\requiv_\Ei$-equivalent to $\play$.  

Formally, the value of $\belief{B_0}{\Ei}(\lambda)$ can be inductively defined as follows: $\belief{B_0}{\Ei}(s_0)=B_0$ and $\belief{B_0}{\Ei}(\lambda \cdot(\action_\Ei,\action_\Ai)\cdot s)=
\updateB{\Ei}(\belief{B_0}{\Ei}(\lambda),\action_\Ei,[s]_{\requiv_\Ei})$
where the function $\updateB{\Ei}:2^\states\times\actions_\Ei\times [\states]_{/_{\requiv_\Ei}}\rightarrow 2^{{\states}}$ is defined by:
\[
\updateB{\Ei}(B,\action_\Evei,[s]_{\requiv_\Ei}) =  \{t \in [s]_{\requiv_\Ei}\mid\exists r\in B,\  \exists\action_\Adami\in \actions_\Adami\text{ s.t. }\ftrans(r,\action_\Evei,\action_\Adami)(t)>0\}.
\]

{

\begin{figure}[htb]
\centering
    \subfloat[Arena of Remark~\ref{ex:belief-smaller-than-observation}]
        {\begin{tikzpicture}[transform shape, scale=1]
\begin{scope}
\tikzstyle{every loop}=[->,shorten >=1pt,looseness=7,]
\tikzstyle{loop top}=[in=55,out=125,loop]
\node[draw,circle,double] (q_0) at (0,0) {$q_{0}$};

\node[draw,circle] (q_f) at (2.5,0) {$q_{1}$};

\draw[->,>=latex] (q_0) to[bend left=25]  node[above][yshift=1mm] {$b\barre b$}
node[above][yshift=5mm] {$a \barre \any$ }(q_f);
\draw[->,>=latex] (q_f) to[bend left=25]  node[below][yshift=-1mm] {$b\barre \any$}
(q_0);
\path (q_0) edge [loop top] node[above] {$b\barre a$} (q_0);
\path (q_f) edge [loop top] node[above] {$a\barre\any$} (q_f);

\draw[dashed] (-0.5,0.5) rectangle (3,-0.5);

\end{scope}
\end{tikzpicture}        \label{fig:ex-2}
}
\hspace{2cm}
\subfloat[Arena of Remark~\ref{rk:beliefNotEnough}]{
\begin{tikzpicture}[transform shape,scale=1]
\begin{scope}
\tikzstyle{every node}=[minimum size=8mm]
\tikzstyle{every loop}=[->,shorten >=1pt,looseness=7,]
\tikzstyle{loop top}=[in=55,out=125,loop]

\node[draw,circle] (s0) at (0,0) {$s_0$};

\node[inner sep=0mm,minimum size=0mm] (int) at (0,-0.75) {};

\node[draw,circle] (s1) at (1.5,-1.5) {$s_2$};
\node[draw,circle] (s2) at (-1.5,-1.5) {$s_1$};
\node[draw,circle] (t1) at (2.5,-3) {$t_2$};
\node[draw,circle,double] (f1) at (0.5,-3) {$f_2$};

\node[draw,circle] (t2) at (-2.5,-3) {$t_1$};
\node[draw,circle,double] (f2) at (-0.5,-3) {$f_1$};
\draw (s0) to node[right][xshift=-2mm] {$\any$} (int);

\draw[->,>=latex,bend right] (int) to node[xshift=-2.5mm,yshift=-1mm] {$\frac{1}{2}$} (s1);
\draw[->,>=latex,bend left] (int) to node[xshift=2.5mm,yshift=-1mm] {$\frac{1}{2}$} (s2);

\draw[->,>=latex,bend left] (s1) to node[below right][yshift=2mm] {$a$} (t1);
\draw[->,>=latex] (t1) to node[left][yshift=-2mm,xshift=2.5mm] {$\any$} (s1);
\draw[->,>=latex,bend right] (s1) to node[below left][yshift=2mm]  {$b$} (f1);
\draw[->,>=latex] (f1) to  node[right][yshift=-2mm,xshift=-2.5mm] {$\any$} (s1);

\draw[->,>=latex,bend left] (s2) to node[below right][yshift=2mm] {$a$} (f2);
\draw[->,>=latex] (f2) to node[left][yshift=-2mm,xshift=2.5mm]  {$\any$} (s2);
\draw[->,>=latex,bend right] (s2) to node[below left][yshift=2mm] {$b$} (t2);
\draw[->,>=latex] (t2) to  node[right][yshift=-2mm,xshift=-2.5mm] {$\any$} (s2);
\end{scope}
\end{tikzpicture}
}
\caption{Arenas and beliefs}\label{fig:rk4}
\end{figure}
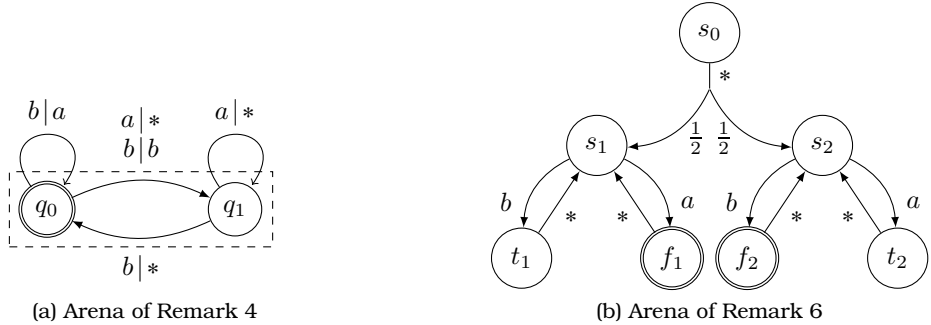

\begin{remark}\label{ex:belief-smaller-than-observation} 
The belief is in general smaller than the currently observed equivalence class. For instance, consider the reachability game depicted in Figure~\ref{fig:ex-2} in which all   states are equivalent. If the strategy of \Eve is to play  $(abb)^{\omega}$, then her observation is always the same (as all states are equivalent). Her initial belief is $\{q_{0},q_{1}\}$  but after playing $a$ it becomes $\{q_{1}\}$ and after a $b$ it  becomes $\{q_{0}\}$ and after another $b$ it becomes  $\{q_{0},q_{1}\}$.
  \end{remark}
}

{
\begin{remark}\label{rk:knowtrans}
Given a family $\mathcal{B} \subseteq 2^S$ of beliefs for \Eve (in the sense that each $B \in \mathcal{B}$ is a subset of a $\requiv_\Ei$-class), one can construct an automaton $\trans_{\mathcal{B}}$ such that the $\trans_{\mathcal{B}}$-compatible strategies are precisely those such that \Eve's belief always remains inside $\mathcal{B}$.
The states of $\trans_{\mathcal{B}}$ are the elements of $\mathcal{B}$, the transition function is defined by $\updateB{\Ei}$, and the actions $\act(B)$ enabled at a state $B$ are those that ensure that the  belief remains inside $\mathcal{B}$.
\end{remark}
}

{\begin{remark}\label{rk:beliefNotEnough}
  In~\cite{GS09,BGG09,BertrandGG17}, it is shown that if \Eve can almost-surely win (using randomised strategies) a B\"uchi game, she can do so using a strategy $\phi$ that only depends on the belief, \emph{i.e.} $\phi(\lambda) = \phi(\lambda')$ whenever $\belief{\phi}{}(\lambda) = \belief{\phi}{}(\lambda')$.  However, even if \Eve is playing alone, this is no longer true\footnote{This fact is also observed in~\cite{CD}.} (even for reachability games) in our setting where we restrict to pure (\ie non-randomised) strategies. Consider the reachability game in Figure~\ref{fig:rk4} where \Eve is playing alone. The equivalence relation is given by $s_{1} \requiv_\Ei s_{2}$, $t_1\requiv_\Ei f_2$ and  $t_2 \requiv_\Ei f_1$.

If the game starts in $s_{0}$ then whatever  strategy \Eve uses, her belief always coincides with her observation. Eve can surely win (she can simply play the sequence $aaab$). But if her  strategy only depends on her belief then she necessarily plays a sequence of actions of the form $xu^\omega$ where $x\in\{a,b\}$ and $u$ is a two-letter word, and thus she has a probability $\frac{1}{2}$ to win using such a strategy.

\end{remark}
}

We now return to the strengthening of Theorem~\ref{theo:semiGame}. We assume that \Adam is perfectly informed and we fix an automaton $\trans = ( Q,\actions_\Ei\times\states_{/_{\requiv_\Ei}},q_0,q_s,\deltaT,\act) $ as in Section~\ref{sec:automaton-driven}. We are interested in checking whether \Eve has a $\trans$-compatible strategy that is positively winning in the reachability game $(\arena,B,F)$. 

Our main result is the following and its proof is by two successive reductions and an application of Theorem~\ref{theo:semiGame}.


\begin{theorem}\label{theo:posTcom}
When \Adam is perfectly informed, one can decide in exponential time in $|S|$ and polynomial in $|Q|$ whether \Eve has a $\trans$-compatible strategy that is positively winning in the reachability game $(\arena,B,F)$. {If such a strategy exists, one can construct one that uses memory of size polynomial in $|Q|$ and exponential in $|S|$.}
\end{theorem}

\begin{proof}
Note that adding the condition on the strategy being $\trans$-compatible somehow means that once a final state is reached the play is not yet won by \Eve because she needs to keep playing in accordance with $\trans$ (\ie she must avoid to produce a partial play $\play$ with $q_\play=q_s$). Hence, it is natural to consider an enriched arena $\arena_\trans$ that embeds $\trans$.
For this let $\arena_\trans = (\states\times Q,\actions_\Evei,\actions_\Adami,\ftrans',\rdoubleequiv_\Ei,\rdoubleequiv_\Ai)$ where
\begin{itemize}
\item $\delta'((s,q),\action_\Ei,\action_\Ai)(s',q')$ equals $\delta(s,\action_\Ei,\action_\Ai)(s')$ if $q'=\deltaT(q,(\action_\Ei,[s']_{\requiv_\Ei}))$ and otherwise it equals $0$;
\item $(s_1,q_1)\rdoubleequiv_\Ei (s_2,q_2)$ if and only if $s_1\requiv_\Ei s_2$ and $q_1=q_2$; and
\item $\rdoubleequiv_\Ai$ is the equality relation, \ie \Adam is perfectly informed.
\end{itemize}

Of special interest is the safety game $(\arena_\trans,B\times\{q_0\},S\times \{q_s\})$ and we are interested in sure winning for \Eve because of the following straightforward lemma 

\begin{lemma}\label{lemma-red1}
\Eve has a (possibly losing) $\trans$-compatible strategy in the reachability game $(\arena,B,F)$ if and only if she has a surely winning strategy in the safety game $(\arena_\trans,B\times\{q_0\},S\times \{q_s\})$.
\end{lemma}

It is a known result~\cite{Berwanger08} that when one considers sure winning for \Eve in a safety game, winning strategies only depend on the belief of  \Eve (in the sense of Section~\ref{sec:automaton-driven}). More precisely consider the (unique) largest subset $\mathcal{B}$ of beliefs and the (unique) mapping $Aut:\mathcal{B}\rightarrow 2^{\actions_\Ei}$ such that the following holds.
\begin{itemize}
\item No belief $B\in\mathcal{B}$ contains a forbidden state.
\item For every $B\in \mathcal{B}$, the set $Aut(B)$ {which} consists of all those actions $\action_\Ei\in Aut(B)$ such that for every action $\action_\Ai\in \actions_\Ai$ one has \oschanged{$\updateB{\Ei}(B,\action_\Ei,[s]_{\requiv_\Ei})\in\mathcal{B}$}, {is not empty}; \ie actions in $Aut(B)$ are those that ensure that the updated belief will still be in $\mathcal{B}$ regardless of the action of \Adam.
\reviewF{1}{why add $\cup {\emptyset}$, while also
 saying that it is not empty?}\answer{Right We removed the $\cup\emptyset$}
\end{itemize}
Then \Eve surely wins the safety game from configurations where her belief $B$ is in $\mathcal{B}$ and a strategy consists in choosing any action in $Aut(B)$.

Note that in the safety game $(\arena_\trans,B\times\{q_0\},S\times \{q_s\})$, \Eve's beliefs are elements in $2^S\times Q$ (as we have that $(s_1,q_1)\rdoubleequiv_\Ei (s_2,q_2)$ implies  $q_1=q_2$).

\reviewF{1}{consider renaming the notion of "belief" to "belief", that way you avoid unpleasant
  sentences like "all beliefs"}\answer{We did that. Also changing all $K$/$\mathcal{K}$ to $B$/$\mathcal{B}$}

Now consider an automaton $\trans'=( Q',\actions_\Ei\times\states_{/_{\requiv_\Ei}},q_0',q_s',\deltaTp,\act') $ that \emph{computes} \Eve's belief (as explained in Remark~\ref{rk:knowtrans}. {Hence, $\trans'$ is the same as $\trans_{\mathcal{B}}$}) in the previous safety game and uses function $Aut=\act'$ to define those authorised actions. To fit the definition, merge all beliefs not in $\mathcal{B}$ in a sink state and define $Aut$ to be equal to $\emptyset$ on it. The states $Q'$ of $\trans'$ are elements of $\mathcal{B}$ (plus the sink state) and one takes as the initial state $q_0'=B\times\{q_0\}$ (which possibly is the sink state). In particular the number of states of $\trans'$ is exponential in $|S|$ and linear in $|Q|$.

Now one can go back to the original arena and consider the enriched arena $\arena_{\trans'}$. Then we have the following easy lemma.

\begin{lemma}\label{lemma-red2}
\Eve has a $\trans$-compatible positively winning strategy in the reachability game $(\arena,B,F)$ if and only if she has a positively winning strategy in the reachability game $(\arena_{\trans'},B\times\{q_0'\},F\times (Q'\setminus\{q_s'\}))$. 

Moreover, from a positively winning strategy in the second game using memory of size $N$ one can effectively construct a  $\trans$-compatible positively winning strategy in the reachability game $(\arena,B,F)$ that uses a memory of size $\mathcal{O}(N\times 2^{|S|}\times |Q|)$.
\end{lemma}

\begin{proof}
If Eve positively wins in $(\arena_{\trans'},B\times\{q_0'\},F\times (Q'\setminus\{q_s'\}))$ then we can safely assume that she necessarily always plays authorised (according to $\act'$) actions (otherwise the play goes directly to $S\times \{q_s'\}$ and gets trap in it forever, hence cannot reach $F\times (Q'\setminus\{q_s'\}$), hence is $\trans$-compatible thanks to Lemma~\ref{lemma-red1}. Such a strategy can be mimicked in the original game and it requires to simulate automaton $\trans'$ hence, costs an extra memory of size the one of $\trans'$. Conversely, if it she has a positively winning $\trans$-compatible strategy in the original game, the same strategy can be mimicked in the reduced game and is still positively winning.
\end{proof}

Now combining Lemma~\ref{lemma-red2} together with Theorem~\ref{theo:semiGame} concludes the proof of Theorem~\ref{theo:posTcom}.
\end{proof}


\subsection{The Case Where  \Adam Is More Informed Than \Eve}\label{subsection:AdamMoreInformed}

\reviewF{1}{explicitly state that when Adam has imperfect information he still observes the actions}\answer{We added a sentence referring to Remark~\ref{rk:action-visible}.}

We now assume that \Adam is more informed than \Eve and we fix an automaton $\trans = ( Q,\actions_\Ei\times\states_{/_{\requiv_\Ei}},q_0,q_s,\deltaT,\act) $ as in Section~\ref{sec:automaton-driven}. Again, we are interested in checking whether \Eve has a $\trans$-compatible strategy that is positively winning in the reachability game $(\arena,B,F)$. The idea here is to reduce this question to one on a game where \Adam is perfectly informed and therefore conclude thanks to Theorem~\ref{theo:posTcom}.

\oschanged{Recall that in this setting, as noted in Remark~\ref{rk:action-visible}, we can safely assume that, against a fixed strategy of \Eve, \Adam observes the actions played by both players.}

For this let $\mathcal{H}$ be all those subsets of $\states$ that consist of $\requiv_\Ai$-equivalent states. 
For such a subset $H$ and for any pair of actions $(\action_\Ei,\action_\Ai)\in(\actions_\Ei\times \actions_\Ai)$ define the set $Up(H,\action_\Ei,\action_\Ai)\in \mathcal{H}$ as follows. First, define $M=\{s'\in S\mid \exists s\in H \text{ s.t. } \ftrans(s,\action_\Ei,\action_\Ai)(s')>0)\}$ as the set of all possible successors of states in $H$ when playing the pair of actions $(\action_\Ei,\action_\Ai)$ and let $Up(H,\action_\Ei,\action_\Ai)$ consist of all those non-empty subsets $H'$ that can be written as $H'=M\cap [s]_{\requiv_\Ai}$, \ie all possible indistinguishable (for \Adam) subsets of $M$.

Define now a new arena $\arena' = (\mathcal{H},\actions_\Evei,\actions_\Adami,\ftrans',\rdoubleequiv_\Ei,\rdoubleequiv_\Ai)$ by letting 
\begin{itemize}
\item $\ftrans'(H,\action_\Ei,\action_\Ai)(H') = 1/|Up(H,\action_\Ei,\action_\Ai)|$ if $H'\in Up(H,\action_\Ei,\action_\Ai)$ and $0$ otherwise;
\item $H_1\rdoubleequiv_\Ei H_2$ if $s_1\requiv_\Ei s_2$ for every $s_1\in H_1$ and $s_2\in H_2$; and
\item $\rdoubleequiv_\Ai$ is the equality relation, \ie \Adam is perfectly informed.
\end{itemize}
Define the set of final states $F'$ as those elements $H$ in $\mathcal{H}$ such that $H\cap F\neq \emptyset$. 

Note that the equivalence classes of $\rdoubleequiv_\Ei$ can be identified with the equivalence classes of $\requiv_\Ei$ (because $\requiv_\Ai \subseteq \requiv_\Ei$) and therefore one can define $\trans$-compatible strategies for \Eve also in a play in $\arena'$. More generally, any \Eve's strategy in one game can be used in the other one.

For a set $B\subseteq \states$ define $\nu(B)\in \mathcal{H}$ as $\nu(B)=\{\{s\}\mid s\in B\}$.
The following proposition relates game $(\arena,B,F)$ and game $(\arena',\nu(B),F')$.

\begin{proposition}\label{prop:moreinftoperfect}
A strategy of \Eve is a positively winning $\trans$-compatible strategy in $\game=(\arena,B,F)$ if and only if it is a positively winning $\trans$-compatible strategy in $\game'=(\arena',\nu(B),F')$. 
\end{proposition}

\begin{proof}
Let $\strat$ be a positively winning $\trans$-compatible strategy in $\game$. Now use $\strat$ in 
$\game'$: obviously it is still $\trans$-compatible and we only have to prove that it is positively winning. Consider a strategy $\psi'$ of \Adam in $\game'$. Then, assuming \Eve respects $\strat$, strategy $\psi'$ can be mimicked in game $\game$: indeed, \Adam simply has to update a state $H$ in $\arena'$ which is done by computing $Up(H,\action_\Ei,\action_\Ai)$ and observing the equivalence class for $\requiv_\Ai$ relation; assuming \Eve respects $\strat$ it means that \Adam always knows what action $\action_\Ei$ she will play and therefore can compute $Up(H,\action_\Ei,\action_\Ai)$. Call $\psi$ the strategy in $\game$ mimicking $\psi'$. 

Now let $N$ be some integer and consider all those partial plays of length $N$ in $\game$ where \Eve respects $\strat$ and \Adam respects $\psi$. \reviewF{1}{Group allÃ¢ÂÂ¦ reword this}\answer{We reworded that.} 
\oschanged{Consider the $\requiv_\Ai$-equivalent classes among these partial plays and for every class consider the set $H$ of possible last states.} 
Then those such $H$ are exactly those states that can be reached in $\game'$ in a partial play of length $N$ when \Eve respects $\strat$ and \Adam respects $\psi'$. As $\strat$ is positively winning in $\game$ , thanks to Proposition~\ref{proposition:positiveWinningReachability} there is some $N$ such that \Eve positively wins within the $N$ first moves and therefore for the same $N$ we conclude that \Eve positively wins within the $N$ first moves in $\game'$ using $\strat$ against $\psi'$. As this property does not depend on $\psi'$ we conclude that $\strat$ is positively winning in $\game'$.

Conversely, assume she has a positively winning $\trans$-compatible strategy in $\game'$. Now use $\strat$ in $\game$: obviously it is still $\trans$-compatible and we only have to prove that it is positively winning. By contradiction, assume \Adam has a strategy $\psi$ that ensures, provided \Eve uses $\strat$ in $\game$, that no final state is reached. Then, from $\psi$ one can define a strategy in $\psi'$ that consists in a partial play $H_0(\action_\Ei^0,\action_\Ai^0)H_1(\action_\Ei^1,\action_\Ai^1)\cdots H_k$ to play action $\psi([s_0]_{\requiv_\Ai}\cdots [s_k]_{\requiv_\Ai})$ where $s_i$ is any (they are all $\requiv_\Ai$-equivalent) element in $H_i$ for every $i$. Using the same argument as in the direct implication relating plays in $\game$ when using strategies $(\strat,\psi)$ and plays in $\game'$ when using strategies $(\strat',\psi')$, one concludes that playing $\psi'$ against $\phi$ in $\game'$ ensures that no final state is visited hence, leading a contradiction with $\phi$ being positively winning in $\game'$.

\end{proof}

Combining Proposition~\ref{prop:moreinftoperfect} with Theorem~\ref{theo:posTcom} directly leads the following result.

\begin{theorem}\label{theo:posTcomMore}
When \Adam is more informed than \Eve, one can decide in double exponential time in $|S|$ and polynomial in $|Q|$ whether \Eve has a $\trans$-compatible strategy that is positively winning in the reachability game $(\arena,B,F)$. {If such a strategy exists, one can construct one that uses memory of size polynomial in $|Q|$ and doubly exponential in $|S|$.}
\end{theorem}

\section{Almost-Surely Winning {for B\"uchi Conditions}}\label{sec:as}

For the rest of this section fix an arena $\arena = (\states,\actions_\Evei,\actions_\Adami,\ftrans,\requiv_\Ei,\requiv_\Ai)$ and a set of final states $F\subseteq S$. {We are interested in almost-sure winning strategies, and} we focus on  B\"uchi conditions, 
as a solution for this case permits to obtain a solution for reachability condition by a simple reduction (change the arena so that whenever a final state is reached then the play stays in it forever). For the moment we do not make any assumption on how \Adam is informed.

We show how to compute the  set of \defin{almost-surely winning beliefs} of \Eve, denoted $\wbelief{}$, which is the  set of subsets $B \subseteq \states$ such that $B \subseteq [s]_{\requiv_\Ei}$  for some $s \in \states$ and for which \Eve has an almost-surely winning strategy in the B\"uchi game $\game_B=(\arena,B,F)$.  \oschanged{For some $B\in \wbelief{}$ we let $[B]_{\requiv_\Ei}=[s]_{\requiv_\Ei}$ for $s\in S$ such that $B \subseteq [s]_{\requiv_\Ei}$.}

\subsection{Fixpoint Characterisation}

Lemma~\ref{theo:BASFixpoint} below states that the set $\wbelief{}$ can be expressed as the greatest fix-point of
a {(monotone)} mapping $\BeliefOperator: 2^{2^\states} \rightarrow
2^{2^\states}$ defined as follows. Let $\mathcal{B}\subseteq
2^\states$ and let $B\in\mathcal{B}$. We say that $B$ belongs to
$\BeliefOperator(\mathcal{B})$  if \Eve has a strategy in the \emph{reachability} game
$(\arena,B,F)$ which is positively winning and guarantees that her
belief always stays in $\mathcal{B}$.

\begin{lemma}\label{theo:BASFixpoint}
$\wbelief{}$ is the greatest fixpoint of $\BeliefOperator$.
\end{lemma}

\begin{proof}
\reviewF{1}{why define $\mathcal{B}$-good when you already did? Instead of saying $B$ is
 $\mathcal{B}$-good, just say $B\in \Xi(\mathcal{B})$}\answer{fixed as suggested}

We first argue that $\wbelief{}$ is a fixpoint for $\BeliefOperator$. For this we consider {any} $B\in\wbelief{}$ and prove that $B\in\Xi(\wbelief{})$. We denote by $\game_B$ the B\"uchi game $(\arena,B,F)$ and we start with a simple lemma.

\reviewF{1}{$\sigma_E=\phi([B]_{\sim_E})$Ã¢ÂÂ¦: $[B]_{\sim_E}$ has not been defined when $B\subseteq S$}\answer{fixed: added a sentence in the last paragraph before section 5.1}

\begin{lemma}\label{lemma:StayInWB}
Let $B\in\wbelief{}$. Let $\phi$ be \emph{any} strategy for \Eve that is almost-surely winning for her in $\game_B$ and let $\action_\Evei = \phi([B]_{\requiv_\Ei})$.
Then, for any $\action_\Adami\in \actions_\Adami$, for any $t$ such that $\exists s\in B$ with $\ftrans(s,\action_\Evei,\action_\Adami)(t)>0$, $\updateB{\Ei}(B,\action_\Evei,[t]_{\requiv_\Ei}) \in\wbelief{}$.
\end{lemma}

\begin{proof}
Consider some action $\action_\Adami$ and some $t$ such that $\ftrans(s,\action_\Evei,\action_\Adami)(t)>0$ and let $B'=\updateB{\Ei}(B,\action_\Evei,[t]_{\requiv_{\Evei}})$. By definition of {$\updateB{\Ei}$}, for every $t'\in B'$, there is some $s'\in B$ and some action $\action^{t'}_\Adami$ such that $\ftrans(s',\action_\Evei,\action^{t'}_\Adami)(t')>0$. Now, define the strategy $\phi'$ of \Eve by letting $\phi'(\lambda) = \phi([s]_{\requiv_\Ei}\cdot \lambda)$ for any partial play $\lambda$. We claim that 
$\phi'$ is almost-surely winning for \Eve in $\game_{t'}$  for any
$t'\in B'$, hence implying that $B'\in \wbelief{}$. By
contradiction, assume that $\phi'$ is not almost-surely winning for
some $\game_{t'}$ with $t'\in B'$ and let $\psi'$ be a
counter-strategy for \Adam in $\game_{t'}$, \emph{i.e.}
$\proba{t'}{\phi'}{\psi'}(\objective)<1$ (recall that $\objective$
denotes here the B\"uchi objective). Now, pick $s'\in B$ such that
$\ftrans(s',\action_\Evei,\action^{t'}_\Adami)(t')>0$ and define a
strategy $\psi$ of \Adam by letting $\psi(s') = \action^{t'}_\Adami$
and $\psi(s'\cdot \lambda) = \psi'(\lambda)$. Then as
$\proba{t'}{\phi'}{\psi'}(\objective)<1$ one also has that
$\proba{s'}{\phi}{\psi}(\objective)<1$ which leads to a contradiction.
\end{proof}

Fix a strategy $\phi_B$ as in Lemma \ref{lemma:StayInWB}: a play $\lambda$ in $\game_B$ where \Eve respects $\phi_B$ is such that $\belief{B}{\Ei}(\lambda)\in\wbelief{}$. 
Moreover, as $\phi_B$ is almost-surely winning for the B\"uchi game $\game_B$, it is in particular positively winning in the \emph{reachability} game $(\arena,B,F)$. Hence, using Proposition \ref{proposition:positiveWinningReachability}, one gets a bound $N_B$ and some $\epsilon_B$, meaning that the probability of a play $\lambda$ in $\game_B$ where \Eve respects $\phi_B$ to visit a final state within its first $N_B$ moves is $\geq\epsilon_B$. Hence, $B\in\Xi(\wbelief{})$, implying that $\wbelief{}$ is a fixpoint for $\BeliefOperator$.

Now we show that any fixpoint of $\BeliefOperator$ is included in $\wbelief{}$. For this assume that ${\BeliefOperator}(\mathcal{B}) = \mathcal{B}$ for some $\mathcal{B}$. As any $B\in \mathcal{B}$ is such that $B\in\Xi(\mathcal{B})$  it comes with some $\phi_B$, $N_B$ and $\epsilon_B$. We let $N = \max{\{N_B\mid B\in \mathcal{B}\}}$ and $\epsilon = \min{\{\epsilon_B\mid B\in \mathcal{B}\}}$. 

Now we define a strategy $\phi$ that consists in playing in rounds of length $N$: at the beginning of some round, \Eve considers her current belief $H$ and plays according to $\phi_H$ in the next $N$ moves; then she restarts with the updated belief, and so on forever. 

Now consider some $B\in \mathcal{B}$. We claim that $\phi$ is
almost-surely winning for \Eve in any in $\game_B$. Indeed, from the
properties of the $\phi_{{H}}$, it follows that any play in $\game_B$ where \Eve respects $\phi$ is such that the belief is in $\mathcal{B}$. Now, as the $\phi_{{H}}$ ensure to visit a final state with probability $\geq \epsilon$ in less than $N$ moves the Borel-Cantelli Lemma implies that $\phi$ is almost-surely winning. Hence, $B\in \wbelief{}$ and this concludes the proof.
\end{proof}

\subsection{Decidability Issues}

As $\BeliefOperator$ is monotone for set inclusion, it suffices to
compute $\wbelief{}$ by successive applications (starting with the
set of all subsets) of the operator $\BeliefOperator$ until
reaching the fixpoint. Since $\wbelief{}\subseteq {2^\states}$, the
fixpoint is reached in at most $2^{|\states|}$ steps. 

Now, as noted in Remark~\ref{rk:knowtrans} the property for a strategy to guarantee that \Eve's belief remains in a set $\mathcal{B}$ can be expressed as the strategy being $\trans_{\mathcal{B}}$-compatible (and the number of states of $\trans_{\mathcal{B}}$ is at most exponential in $|\states|$). Therefore, thanks to Theorem~\ref{theo:posTcom} (\resp Theorem~\ref{theo:posTcomMore}) every step in the fixpoint computation can be achieved in time exponential (\resp doubly exponential) in $|S|$ if \Adam is perfectly informed (\resp more informed than \Eve). This leads the following result.

\begin{theorem}\label{thm:as_main}
Let $\game$ be a B\"uchi ({or} reachability) game with $n$ states. 
\begin{itemize}
\item If \Adam is perfectly informed, one can decide whether \Eve has an almost-surely winning strategy in time exponential in $n$. If such a strategy exists, it can be effectively constructed and requires memory at most exponential in $n$.
\item If \Adam is more informed than \Eve, one can decide whether \Eve has an almost-surely winning strategy in time doubly exponential in $n$. If such a strategy exists, it can be effectively constructed and requires memory at most doubly exponential in $n$.
\end{itemize}
\end{theorem}

\begin{proof}
Decidability follows from Theorem~\ref{theo:posTcom}/Theorem~\ref{theo:posTcomMore} and the fixpoint characterisation given in Lemma~\ref{theo:BASFixpoint}. The result on the strategies is also a consequence of Theorem~\ref{theo:posTcom}/Theorem~\ref{theo:posTcomMore} combined with Corollary~\ref{cor:boundingN} which permits to bound the size of $N$ in the proof of Lemma~\ref{theo:BASFixpoint}.
 \end{proof}


\section{Lower Bound}

\reviewF{1}{this paragraph is a little bit hard to read, it can benefit from some rewording}\reviewF{2}{the crucial sentence describing the difference in the coding is very long and difficult to parse, moreover "Adam's was" should be corrected. Please reformulate.}\answer{We hope it is better now}

We now give a matching lower bound to the upper bounds in Theorem~\ref{theo:posTcomMore} and in Theorem~\ref{thm:as_main} for the case where \Adam is more informed than \Eve. Note that in the case where \Adam is perfectly informed one can get a matching lower bound (ExpTime-hardness) as in the case where randomised strategies are allowed~\cite{CDHR07}.  {Also note that in the case where \Adam is more informed than \Eve similar lower bounds, when randomised strategies are allowed, were obtained in~\cite{GS09,BGG09,BertrandGG17} for almost-sure winning\footnote{Actually the lower bound in~\cite{GS09} uses a game where none of the player is more informed than the other but it is easily seen how to modify it to obtain a game where \Adam is more informed than \Eve}; however, even if the ideas of the proof below are similar to the ones in~\cite{GS09,BGG09,BertrandGG17}, namely the players simulate a run of an exponential space alternating Turing machine while gadgets prevent cheating, there is a slight but crucial difference. Indeed, \Eve is in charge of describing the successive configuration while in previous proofs \Adam was; we actually believe that this is needed for the proof to work (mainly because we are not only interested in almost-sure winning but also in positively winning).}

\begin{theorem}\label{theo:lowerbound}
Deciding whether \Eve has a positively winning (\resp an almost-surely winning) strategy in a reachability game where \Adam is more informed than her is a 2-ExpTime-hard problem.
\end{theorem}

\begin{proof}
The idea is to simulate a computation of an alternating Turing machine that uses a space of exponential size and to reduce termination to almost-surely winning for \Eve. As alternating Turing machines of exponential space are equivalent to deterministic Turing machines working in doubly exponential time it permits to obtain the desired lower bound. \oschanged{We can safely assume that the tape alphabet $A$ contains a blank symbol as well as a special symbol $\sharp$, and that initially the input tape is made of $n$ successive $\sharp$ symbols followed by $2^n-n$ blank symbols.}\reviewF{1}{what is the
 meaning of the word distinguished here? Different from the blank symbol?}\answer{It should be clear now} A configuration of the machine can be described by a word of length $2^n$ in $A^*QA^*$ where  $Q$ is the set of states of the machine (including some final states): the meaning of a configuration $a_1\cdots a_{\ell} q a_{\ell+1}\cdots a_{2^n}$ is that the tape content is $a_1\cdots a_{\ell} a_{2^n}$, the state is $q$ and the reading/writing head is on the $\ell$-th cell. A run of the machine is a sequence of successive configurations separated by transitions of the machine; it is accepting if it contains a final configuration (and in that case the run is of finite length; otherwise it is of infinite length).

A classical way of thinking of an alternating Turing machine is as a game where \Eve is in charge of the choice of transitions when the machine is in an existential state while \Adam takes care of the universal states. The machine accepts if and only if \Eve has a winning strategy to eventually reach a configuration with a final control state.

Consider now the following (informal) game. \Eve is in charge of describing the run of the Turing machine (her actions' alphabet contains all the necessary symbols for that \ie $A\cup Q$ that permits the game to go in some associated states). After she described a configuration either she (in case the state is existential) or \Adam (in case the state is universal) describes a valid transition of the machine (again by playing some special actions), and then \Eve describes the successive configuration and so on until possibly a final configuration is reached (in which case she wins the game). \oschanged{Hence, in the game's state one stores the state (of the Turing machine) when described as well as the adjacent symbols; this information is used when \Eve/\Adam has to describe the next transition of the machine (that should be consistent).}

Of course the problem is that \Eve could cheat by not describing a valid run. For this, \Adam can, in every configuration, secretly (\ie \Eve does not observe it) mark a cell of the tape, and in the next configuration he can indicate a cell (supposedly of same index than the previously marked one) and it is checked whether it has been wrongly updated: this is easily done as the cell before and after the marked cell have been stored in the arena (and \Eve does not observe it of course) and together with the transition one can compute the correct update of the cell. Now in case there is indeed a wrong update of the cell content, the play restarts (\ie the players restart from the initial configuration of the Turing machine); otherwise the play goes to a final state and \Eve wins. \oschanged{Hence, in the game's state one also stores (and hides to \Eve) information of a cell marked by \Adam and of the adjacent symbols for a later check. This marking by \Adam as well as the checking later is done by him playing a distinguished action.}

One problem in the previous simulation is that \Adam could cheat by indicating two cells that are not with the same index. If the space used by the machine was of linear size, one could of course store the actual index and formally check it. Here, we use an extra coding to circumvent this problem. When describing the configuration, after every symbol \Eve produces a sequence of $n$ bits whose meaning is to describe, in binary counting, the index of the last symbol. When she describes such a binary number, \Adam can secretly mark a bit that he claims will be not correctly updated when describing the index of the next symbol (for this he just plays an action that stands for a number between $1$ and $n$) and this is checked next: if she made an incorrect update, the play restarts (\ie the players restart from the initial configuration of the Turing machine); otherwise the play goes to a final state where she wins. \oschanged{Hence, in the game's state one also stores (and hides to \Eve) the index (between $1$ and $n$) of a bit marked by \Adam as well as the value corresponding to the bit of the same index in the incremented version of the described number (this can be computed on the fly). This marking by \Adam is done by him playing a distinguished action; the checking is done deterministically (thanks to a counter).}
One also uses this binary encoding of the index of the cell in the following way: whenever \Adam marks a symbol that he claims will be incorrectly updated in the next configuration, a bit of its binary encoding is guessed (\ie randomly chosen) and its index is stored and not observed by none of the players. Later, when \Adam indicates the supposed corresponding symbol in the next configuration, the guessed bit is checked and should match: if not the play goes to a final state and \Eve wins; otherwise one does as previously explained (\ie one checks whether the symbol is correct: if not the play restarts otherwise the play goes to a final state and \Eve wins). \oschanged{Hence, in the game's state one also stores (and hides to both players) the value and index of the randomly chosen bit.}

\reviewF{1}{Proof of Theorem 6.1: while it is standard for such proofs to be more informal, here it makes it a little hard to
 follow. The reader would benefit greatly e.g. from knowing what the states of the game are, perhaps not explicitly but
 at least in more detail. }\answer{We tried in the above description to clarify things but, as the reviewer does, we believe that giving an explicit description would be very technicalÃ¢ÂÂ¦}

\reviewF{1}{Next two paragraphs: add a few commas}\answer{We reworked it}

We claim that \Eve positively wins (equivalently almost-surely wins) this game if and only if the Turing machine accepts. Once this is established, the proof will be over, as one can easily notice that the previous informal game can be encoded formally as a two-player game with imperfect information of polynomial size in the one of the Turing machine.

First, assume that the Turing machine accepts. Hence, it means that the existential player \Eve has a winning strategy in the acceptance game of the machine. Now, mimic this strategy in the above described game: \Eve always makes a correct description of a run and, when she has to choose a transition of the machine, she does as in her winning strategy in the acceptance game of the machine. We claim that this strategy is almost-surely winning (hence, also positively winning). Indeed, any strategy of \Adam that does not infinitely often claim that a cell is incorrectly updated is surely losing for him because either he makes a wrong claim (actually his claims are always wrong but here we mean he get discovered because of the hidden bit), or after some point the simulation goes to the end and finishes by a final configuration of the Turing machine. Now, against this strategy of \Eve, when \Adam infinitely often claims that a cell is incorrectly updated, he almost-surely gets caught because at every claim there is a (fixed positive) probability (at least $1/n$) that the secret bit does not match, and therefore, by Borel-Cantelli Lemma, the probability that he gets caught eventually is $1$. Of course, if \Adam claims at some point that a bit is incorrectly updated by \Eve he also looses (because she describes a valid run). Hence, \Eve's strategy almost-surely defeats any strategy of \Adam.

Conversely, assume that the Turing machine does not accept. Hence, it means that the existential player \Eve has no winning strategy in the acceptance game of the machine. 
Now, consider a strategy of \Eve.
There are two possibilities. 
\begin{itemize}
\item Either there is a strategy\footnote{In fact a set of indistinguishable strategies from \Eve's point of view, including the ones where \Adam claims she cheats.} of \Adam against which \Eve's strategy eventually cheats. Then, consider the strategy of \Adam that plays the same except that he points the moment where she cheats: then, \Eve must behave the same and therefore the play restarts. Now, consider how she behaves in the restarted play and do the same reasoning. If we are always in the same situation, by iteratively playing a strategy pointing where she cheats in every simulation of the Turing machine ensures that no final configuration is reached and therefore that she surely looses. 
\item Or, against any strategy of \Adam, \Eve's strategy never cheats (\ie describes a valid run). Hence, \Eve's strategy can be seen as a strategy in the acceptance game of the machine and therefore, one can consider the strategy of the universal player that beats it in the acceptance game and let \Adam mimic it in the simulation game (and he never claims that she cheats). Then, this strategy leads to an infinite play that corresponds to the description of an infinite run of the alternating Turing machine that never visits a final configuration: hence, it surely defeats \Eve's strategy
\end{itemize}

In conclusion, for any strategy of \Eve in the above described game there is a strategy of \Adam that surely beats this strategy, which implies that there is no positively winning (hence almost-surely winning) strategy for \Eve. 
This terminates the proof.
\end{proof}


\section{Summary}

\reviewF{2}{I would like to suggest to put a small description of the model also
in the conclusions, so that they can be used for a quick reference to
the results without the need to read the whole introduction.
}\answer{We added the next two paragraphs}

\oschanged{
In this paper we considered finite state games in which, at each round,
the two players (called \Eve and \Adam) choose concurrently an action and based on these
actions the successor state is chosen according to some fixed
probability distribution. We considered several classical winning conditions: safety, reachability, B\"uchi and co-B\"uchi.  Moreover, the players are imperfectly informed: each player has an equivalence relation over states and, instead of observing the exact state, he observes its equivalence class. 
Finally, we restricted our attention to pure strategies, \emph{i.e.} we forbid the players to randomise when choosing their actions.  

We studied the decidability and complexity status of the problem of deciding whether \Eve has a positively (\resp almost-surely) winning strategy. To obtain positive results, we imposed restrictions on how \Adam is informed: we considered the case where he has perfect information and the case where he is more informed than \Eve.
}

The landscape of decidability and undecidability results with pointers
to the literature and to the results in our paper are shown in
the Table~\ref{table:summary}. The entries of the form
``1/2-Exptime-complete''\ refer to the two cases of \Adam being perfectly
informed and being better informed than \Eve, respectively (the
results from~\cite{CD} are for the case of \Adam being perfectly
informed). The implication $\Rightarrow$ means that our result is an
easy consequence of a result from the literature. The undecidability
results already hold for the case in which \Adam is perfectly
informed.

\begin{table}[htb]
\begin{center}
{\footnotesize
\begin{tabular}{|c||c|c|c|c|}\hline
             & Safety             & Reachability             & B\"uchi                   & co-B\"uchi                        \\ \hline \hline
 Positively  & Undecidable        & 1/2-Exptime-complete          & Undecidable               & Undecidable                       \\ 
             &Th.~\ref{theo:undecidable-positive-safety} & ~\cite{CD}, Th.~\ref{theo:semiGame}/Th.~\ref{theo:posTcomMore}  + Th.~\ref{theo:lowerbound}                          &   ~\cite{BBG08} $\Rightarrow$ Th.~\ref{Th:undecidabilityPOMDP}                       & Th.~\ref{theo:undecidable-positive-safety}  \\\hline
 Almost & ExpTime-comp.   & 1/2-Exptime-complete          & 1/2-Exptime-complete          & Undecidable                       \\
       Sure &~\cite{Berwanger08} &~\cite{CD}, Th.~\ref{thm:as_main} + Th.~\ref{theo:lowerbound} & Th.~\ref{thm:as_main} + Th.~\ref{theo:lowerbound} &~\cite{BBG08} $\Rightarrow$ Th.~\ref{Th:undecidabilityPOMDP}  \\ \hline
\end{tabular}}
 \caption{Landscape of decidability and undecidability results}
 \label{table:summary}
 \end{center}
 \end{table}



\newcommand{\noopsort}[1]{} \newcommand{\singleletter}[1]{#1}
  \newcommand{\etal}{et al.}

\end{document}